\documentclass[a4paper,12pt]{article}
\usepackage{amsmath,amssymb,latexsym,amstext,amsthm}
\usepackage{graphicx,subfig,sidecap}
\usepackage{enumerate}
\usepackage{times}
\usepackage{float}
\usepackage{natbib}
\usepackage{hyperref}
\usepackage{authblk}  
\usepackage{soul} 
\usepackage{xcolor} 

\definecolor{azulUC3M}{RGB}{0,0,102}
\definecolor{gray97}{gray}{.97}
\definecolor{gray75}{gray}{.75}
\definecolor{gray45}{gray}{.45}
\usepackage{listings}
\lstnewenvironment{listing}[1][]
   {\lstset{#1}\pagebreak[0]}{\pagebreak[0]}
 
\lstdefinestyle{consola}
   {basicstyle=\scriptsize\bf\ttfamily,
    backgroundcolor=\color{gray75},
   }
 
\lstdefinestyle{C}
   {language=C,
   }
\usepackage[top=2cm]{geometry}
\newtheorem{proposition}{Proposition}
\newtheorem{remark}{Remark}

\newcommand{\I}{\mathcal{I}}
\newcommand{\C}{\mathcal{C}}
\newcommand{\W}{\mathcal{W}}

\newcommand{\E}{\mathbb{E}}
\newcommand{\R}{\mathbb{R}}
\usepackage[normalem]{ulem}
\usepackage{cancel}
\usepackage{hyperref}

\title{Depth-based estimation for multivariate functional data with phase variability}

\author[1]{Ana Arribas-Gil}
\author[2]{Sara López-Pintado,}
\affil[1]{{\small Departmento de Estadística, Universidad Carlos III de Madrid, Getafe, Spain}}
\affil[2]{{\small Department of Health Sciences,  Northeastern University, Boston,USA}}

\begin{document}
\maketitle

\begin{abstract}
In the context of multivariate functional data with individual phase variation, we develop a robust depth-based approach to estimate the main pattern function when cross-component time warping is also present. In particular, we consider the latent deformation model (Carroll and M\"{u}ller, 2023) in which the different components of a multivariate functional variable are also time-distorted versions of a common template function. Rather than focusing on a particular functional depth measure, we discuss the necessary conditions on a depth function to be able to provide a consistent estimation of the central pattern, considering different model assumptions. We evaluate the method performance and its robustness against atypical observations and violations of the model assumptions through simulations, and illustrate its use on two real data sets. 
\end{abstract}
\noindent%

\section{Introduction}\label{intro}
Functional data is increasingly being collected in many fields of science, such as medicine, biology, and economics. Statistical methods have been developed in the last several decades to address different problems, especially in the univariate functional data analysis (FDA) setting (see Ramsay and Silverman, 2005, Ferraty and Vieu, 2006, Wang et al. 2016). Some of these methodological developments are functional principal component analysis \citep{FPCA, FPCA2}, functional regression \citep{Freg, Freg2, Freg3}, and functional classification and clustering \citep{Fclass,Fclass2,Fclus,kmeans_alig}. A particular aspect of functional data is the presence of different types of variation, such as phase and amplitude variation \citep{PhaseAmp}. In the presence of phase variation, also known as \emph{time warping}, the objective is to unmask an underlying common pattern by aligning or registering the observed curves. A usual model is one in which observed curves are the composition of this unknown common pattern with individual \emph{warping} functions or time distortions. 
Different approaches to curve alignment include landmark methods \citep{landmark1}, continuous monotone registration \citep{RamLi}, moments registration \citep{moment_reg}, pairwise registration \citep{TangMuller_Bio} or  the Fisher–Rao metric framework \citep{FisherRao, FisherRao2}, all of which ultimately rely on solving a minimization problem, while Bayesian resampling registration aims at retrieving the posterior distribution of the warping functions \citep{BayReg}. All of these strategies work by first estimating the warping functions, then aligning the curves, and finally estimating the common underlying pattern from the registered sample. However, an efficient alternative that avoids registration consists of estimating the unknown common pattern directly from the observed sample, as proposed in \cite{Loubes, RobustTW}. In particular, the latter is based on functional depth measures, ensuring robustness of the estimator in the presence of atypical observations. \\
The notion of data depth for functional data is a powerful nonparametric tool that has been extensively studied in the last few decades. It measures the centrality or representativeness of an observation within a distribution or sample. Therefore, it provides ways of ranking functional data from center-outward and of extending robust statistics, such as the median, trimmed means, or rank tests to univariate and multivariate functional data. A notion of data depth needs to satisfy desirable properties, such as maximality at center and monotonicity with respect to the center, non-degeneracy and different types of invariance. For an overview of different functional depth measures and the analysis of their properties, see \cite{NietoReyesBattey,GijbelsNagy, IntegratedDepth}. 

In this paper, when we consider the multivariate functional framework, we deal with multiple functional variables observed per individual which might also present phase variability. In particular, two sources of phase variability are now possible, between individuals and between components of the multivariate process, which might present a common underlying pattern. This is the case, for instance, in growth analysis where the growth of different physiological structures is monitored for a set of individuals \citep{Carroll}. This situation also arises in the analysis of multivariate distributional data, in which the different marginal distributions belong to the same parametric family \citep{distributional}. 

In this paper, we focus on the \emph{latent deformation model}, denoted LDM and introduced in \cite{latentDef}, where each multivariate function are time-distorted versions of a common template, and in addition each subject has its own phase and magnitude variation. The individual time warping is assumed to be the same for each subject across all components, which makes the overall warping \emph{separable} into its individual and component-based origin allowing to estimate each source of phase variability separately. \cite{latentDef} proposed consistent estimators for all elements of the model, based on registration methods across components and in the pooled global sample.

Our contribution is two-fold. First, we propose an alternative estimation model in the LDM that relies in the notion of functional depth and deepest function and generalises the approach of \cite{RobustTW} to the multivariate setting, allowing for a robust and efficient registration-free procedure. Indeed, instead of first solving multiple minimization problems to retrieve the warping functions and then estimate the common underlying pattern, the depth-based approach allows for direct estimation of this target, from which we can then obtain the individual and component based time distortions. Second, we study the necessary properties that a functional depth measure needs to satisfy to guarantee consistent estimation of the common target function in the considered time-warping framework. In particular, we extend the results of \cite{RobustTW} for the modified band depth \citep{BandDepth} to a broader class of functional depths.

The rest of the article is organised as follows: In Section \ref{sec:model} the LDM is presented. In Section \ref{est} we develop our proposal, by first reviewing the univariate time-warping model and establishing the conditions that a functional depth needs to satisfy to provide consistent estimation. Then we propose a depth-based estimation approach for the multivariate LDM, deriving its consistency under certain assumptions on the individual and component- based warping processes. We also develop a graphical tool, the WHyRA (Warping Hypograph Ranking Agreement) plot, as an exploratory diagnostic tool to help visualize and check the hypothesis of the LDM. In particular it is designed to validate the preservation of the individual warping functions across components, from the warping estimates. In Section \ref{sim_sect} we conduct an extensive simulation study in which we investigate the small sample performance of our method and its robustness under violations of the assumptions of the LDM and in the presence of outlying trajectories. We also compare the results with those obtained with the method proposed in \cite{latentDef}. In Section \ref{appli} we demonstrate how our procedure can be used to analyse multivariate functional real data for which inter-individual and inter-component phase variability is present. In particular, we study data on ice extent on the Artic Ocean over time and maternity-age distributions across Europe, and illustrate the relevance of the proposed methodology to identify atypical observations or detect distortions from the assumptions of the LDM. We finalize the article with a discussion in Section \ref{disc}.

\section{Latent deformation model}\label{sec:model}
In this section we present the latent deformation model for multivariate functional registration, as proposed in \cite{latentDef}.

Let $\bf{ X} = (X_1,\ldots,X_p)$ be a $p$-dimensional random process with values on $L_2({\I})^p$, where ${\I}=[T_1, T_2]$ is a real interval. Consider a random sample of this process, $\{\bf{ X}_i\}_{i=1}^n$, $\bf{ X}_i(t)= (X_{i1}(t), \ldots, X_{ip}(t))$. The latent deformation model assumes a common structure among components of the p-dimensional curves and individual time-warping such that

\begin{equation}\label{model}
X_{ij}(t) = a_{ij} \lambda (g_{ij}(t)) = a_{ij} \lambda (\psi_j (h_i (t))) =  a_{ij} (\lambda \circ \psi_j \circ h_i) (t),
\end{equation}
for $i=1,\ldots,n,\, j=1\ldots,p$, where
\begin{itemize}
\item $\lambda$ is an unknown deterministic function in $\C(\I)$, referred to as \emph{latent curve} or \emph{amplitude function}
\item $a_{ij}$ are i.i.d. realizations of a real-valued scale variable $A$
\item $\psi_j$ are unknown deterministic component-based distortion functions 
\item $h_i$ are i.i.d. realizations of a time-warping process $H$.
\end{itemize}
Both distortion and time-warping functions belong to the convex space ${\W_{{\I}}}=\{g:{\I}\rightarrow {\I} \,|\, g\in {\C}^2({\I}),\, g(T_1)=T_1,\,g(T_2)=T_2\}$ and so does the composition of any pair of them, $g_{ij}(t) = (\psi_j \circ h_i) (t)$, $i=1,\ldots,n,\, j=1\ldots,p$. \\
Notice that the component-based distortion functions $\psi_j\in{\W_{{\I}}}$, $j=1,\ldots,p$, are considered to be deterministic, and so is the unknown amplitude function $\lambda\in\C(\I)$.

Given an observed sample of multivariate functions generated from (\ref{model}), the aim is to recover the unknown amplitude function $\lambda$, which accounts for the latent structure of the process. It is also of interest to estimate the component-based distortions  $\psi_j$, which allows to estimate each component's main pattern, $\gamma_j=\lambda\circ \psi_j$. Finally, the estimation of the individual warping functions, $h_i$, allows to align or register the observed curves in each of the components of the process. Notice that the way in which the overall time distortion in model (\ref{model}) can be decomposed in the two distortion elements  $\Psi_j$ and $h_i$, is referred to as \emph{separability} in \cite{latentDef}.

An example of curves generated under this model is presented in Figure \ref{fig_model}.
\begin{figure}
\begin{center}
\includegraphics[width=14cm]{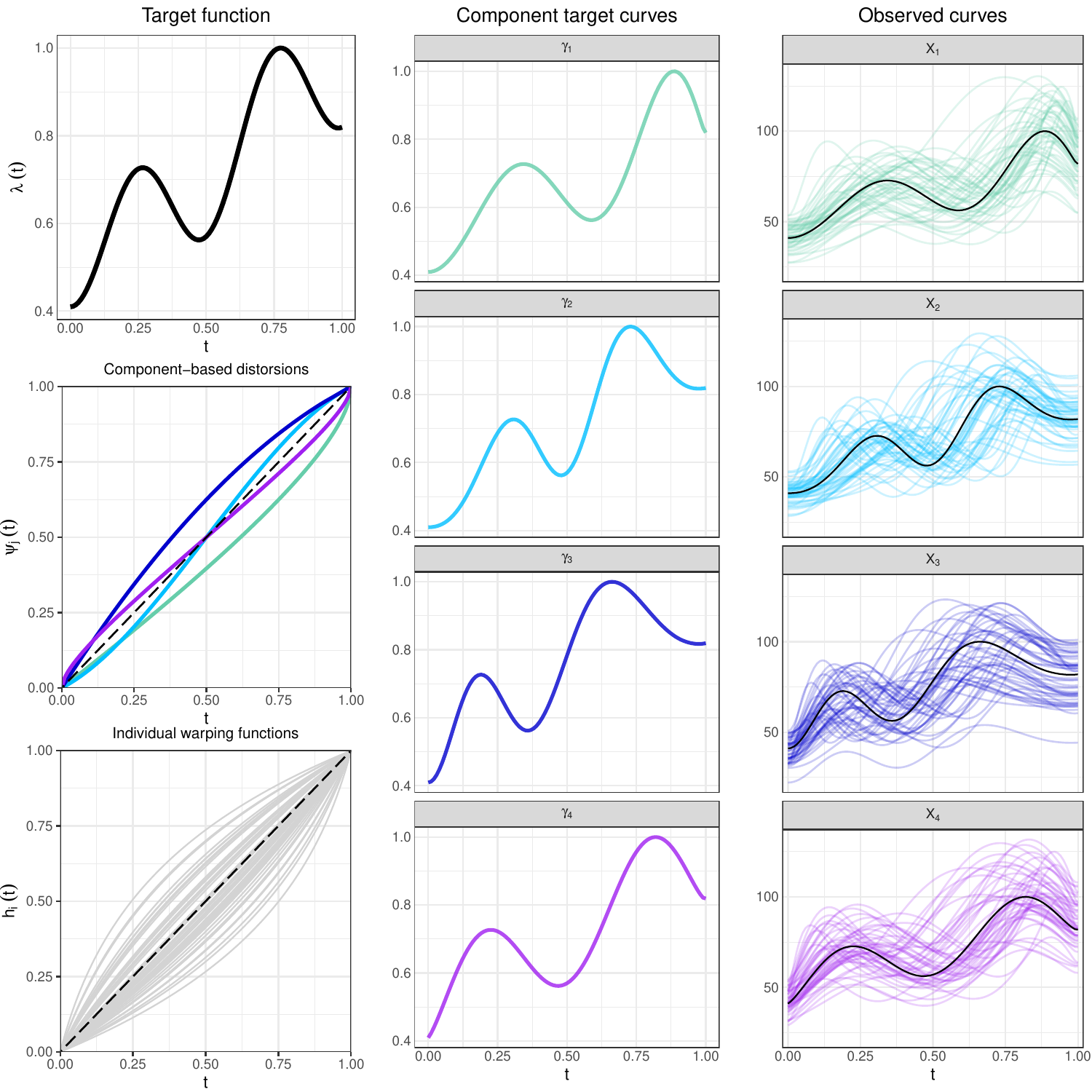}
\end{center}
\caption{Illustration of the latent deformation model with $p=4$ and $n=50$. Left: target function (top) and component-based distortion functions ($\psi_j$'s). Center: component target functions ($\gamma_j$'s). Right: Observed data with component target functions as a black solid line.}\label{fig_model}
\end{figure}

\section{Depth-based medians in the latent deformation model}\label{est}
\subsection{Univariate functional case}
An alternative registration-free estimation method in the univariate \emph{time-warping} model is based on the use of functional depth-based medians, as proposed in \cite{RobustTW}. The authors showed that the \emph{modified band depth (MBD)} for functional data \citep{BandDepth} provides a consistent way of estimating the target function in a common amplitude time-warping model. Indeed, assume that observed functions are given by 
\begin{equation}\label{univ_TW} 
X_i(t)=\lambda\circ h_i(t)=\lambda(h_i(t)), \qquad t\in \I,\quad i=1,\dots,n
\end{equation}
where $h_i$ are i.i.d. realizations of a general warping process $H$, and $\lambda$ is a deterministic unknown amplitude function in $\C(\I)$. Then, the \emph{modified band depth median} of the observed curves, $\hat{m}^{MBD}_{X_{1:n}}$, which is the observation with highest modified band depth within the sample, is a consistent estimator of $\lambda$, as long as $\lambda$ is strictly monotone. This comes from the fact that modified band depth is preserved through composition with strictly increasing functions, so $\hat{m}^{MBD}_{X_{1:n}} =  \hat{m}^{MBD}_{\lambda\circ h_{1:n}}= \lambda \circ \hat{m}^{MBD}_{h_{1:n}}$ (\citep{BandDepth, RobustTW}). If we assume that the population modified band depth median of the warping process $H$ is the identity function on $\I$, that is $m^{MBD}_H (t) = t$, $t\in\I$, then the convergence of $\hat{m}^{MBD}_{\{X_{1:n}\}}$ towards $\lambda$ follows from the consistency properties of the sample modified band depth median towards its population counterpart. Notice that $m^{MBD}_H =id_{\I}$ is a natural assumption equivalent to common time-warping model identifiability assumption $E[H^{-1}]=id_{\I}$, where $id_{\I}$ denotes the identity function on $\I$, $id_{\I}(t)=t$, $t\in\I$ \citep[see][]{RobustTW,latentDef}.

The use of a functional depth-based median in the time-warping model provides a simple way of estimating the target function without having to register the observed time distorted curves. In addition, it has been shown to be a more robust estimation method against atypical curves than those approaches based on mean estimation \citep{RobustTW}.

In fact, the result on the consistency of the modified band depth median with respect to the target/amplitude function $\lambda$ in the univariate common-amplitude model (\ref{univ_TW}) can be extended to any functional depth measure that is invariant under strictly monotone transformations, as stated in the following proposition. The proof is provided in the Appendix.

In what follows we will denote by ${\cal P}(S)$ the set of all probability measures on an arbitrary measurable spaces $S$ equipped with some fixed $\sigma$-algebra. For a random variable $X$ on $S$, we denote by $P_X$ its probability distribution.

\begin{proposition}\label{prop_univ}
    Let $FD: L_2({\I}) \times {\cal P} (L_2({\I}))\longrightarrow \R:\, (x,P)\mapsto FD(x,P)$ be a functional depth measure with empirical version denoted by $FD_{X_1,\ldots,X_n}(x)$. Let $X$ denote a general random variable with values in $L_2({\I})$ and distribution $P_X$. Also, define the sample median as $\hat{m}^{FD}_{X_{1:n}}=\arg\max_{X_i} FD_{X_1,\ldots,X_n}(X_i)$, and the population median, $m^{FD}_X=\arg\max_{X\in L_2}FD(\cdot,P_X)$. If $FD$ satisfies
\begin{enumerate}
\item[P1.] For any continuous strictly monotone function $g:\R \longrightarrow \R$ and any $x\in  L_2({\I})$, 
$$FD(x,P_X) = FD(g(x),P_{g(X)}).$$ 
\item[P2.] Consistency of the sample median, $\hat{m}^{FD}_n$, to the population median, $m^{FD}_X$. That is, under some conditions on $P_X$, $\hat{m}^{FD}_{X_{1:n}} \stackrel{a.s.}{\longrightarrow} m^{FD}_X$ as $n\longrightarrow \infty$.
\end{enumerate}
Then, under the same conditions on $P_X$ required in (P2), for a random sample $X_1,\ldots,X_n$ generated under model (\ref{univ_TW}) and satisfying $m_H^{FD}=id_{\I}$,
$$\hat{m}^{FD}_{X_{1:n}}\stackrel{a.s.}{\longrightarrow} \lambda \qquad \mbox{ as } \quad n\longrightarrow \infty$$
for $\lambda$ strictly increasing and continuous.
\end{proposition}

Many well-known univariate functional depths, such as the band depth and modified band depth, satisfy conditions (P1) and (P2) from Proposition 1 (see \citep{BandDepth}).  In particular, the consistency property (P2) of the median function is a standard property that most notions of functional depth satisfy \citep[see, for instance][]{NietoReyesBattey,GijbelsNagy}. However, invariance with respect to strictly monotone transformations has not been broadly addressed, but holds for many types of functional depths (including integrated and non-integrated depths). Note that the well known integrated and infimal types of functional depths are defined in terms of the distribution of the univariate depth measures on the function values over the time domain, calculated with respect to the marginal probability distributions. In these settings, the invariance property can be inherited. In particular, it can be seen that for any integrated or infimal depth which are based on either integrating or taking the infimal of univariate depths at each time $t$, if the univariate depth considered is based on ranks, such as, univariate halfspace or simplicial depth, and not on distances (Mahalanobis depth) then the monotonicity property (P1) is trivially satisfied.

Indeed, let  $D: \R \times {\cal P} (\R)\longrightarrow \R:\, (x,P)\mapsto D(x,P)$ denote a univariate depth measure. A general functional integrated depth is defined from $D$ as the Lebesgue integral of the univariate depth values over the functional domain. That is, $FD_D(x,P_X)=\int_{\cal I} D(x(t),P_{X(t)}) dt$ where  $P_{X(t)}$ denotes the marginal probability distribution of $X$ at $t\in {\cal I}$.
Then, as long as the univariate depth measure $D$ is invariant with respect to strictly monotone transformations, the integrated functional depth will be invariant too (see \citep{IntegratedDepth} for a discussion of other invariance properties for integrated functional depths).\\
Some non-integrated functional depths can also be characterized in terms of the values of a univariate depth measure on the functional values over the time domain. The infimal depth (or $\Phi$-depth) \citep{Mosler2013}, for instance, is defined as the infimum of such values. Again, if the univariate depth measure is invariant with respect to strictly monotone transformations, so is the associated infimal depth. In this sense, the recently proposed family of \emph{quantile integrated depths, QID}\citep{QID}, which is obtained by integrating up to the K-th quantile of the univariate depths, also satisfies this property. Note that by definition QID includes both integrated and infimal depths.\\
Another functional measure for which this is true is the extremal depth \citep{NarisettyNair}, since it is defined in terms of the cdf of the univariate depths on the function values over the time domain. 

Consistency of the functional median to the amplitude function in the time-warping model (\ref{univ_TW}) can be extended to the general case of a non strictly monotone function $\lambda$ by means of a monotonizing transformation of the sample curves \citep{Loubes}. In fact, if $T: {\cal C}({\cal I}) \longrightarrow {\cal M}({\cal I})$ is an operator taking a general continuous function into a strictly monotone function such that $T(x\circ h)=T(x) \circ h$ for any $x\in \C(\I)$ and any $h\in \W(\I)$, we say that $T$ \emph{preserves warping functions} and the following result holds.
\begin{proposition}\label{prop_univ_gen}
Let $X_1,\ldots,X_n$ be a random sample generated under model (\ref{univ_TW}), with $\lambda$ a general continuous function on $\I$. Let $FD$ be a functional depth measure satisfying assumptions P1) and P2) in Proposition \ref{prop_univ} and for which $m_H^{FD}=id_{\I}$, and let $T: {\cal C}({\cal I}) \longrightarrow {\cal M}({\cal I})$ be a monotonizing operator preserving warping functions.\\
Then, under the same conditions on $P_X$ required to establish the consistency of $\hat{m}^{FD}_n$,
$$\arg\max_{X_i} FD_{T(X_1),\ldots,T(X_n)}(T(X_i))\stackrel{a.s.}{\longrightarrow} \lambda \qquad \mbox{ as } \quad n\longrightarrow \infty.$$
\end{proposition}

The previous result, which proof is given in the Appendix, establishes that for a general amplitude function $\lambda$, the sample curve whose monotonized transformation corresponds to the functional median in the monotonized sample, provides a consistent estimator of $\lambda$. This result had already been proven in \cite{RobustTW} for the modified band depth and is now extended to a general functional depth measure under the conditions of Proposition \ref{prop_univ}. Details on a specific warping preserving monotonizing operator $T$ can be found in \cite{Loubes} and \cite{RobustTW}.

\subsection{Robust estimation in the latent deformation model}\label{estimators}
Now, in the multivariate functional case, we consider, without loss of generality, the common amplitude latent deformation model,
\begin{equation}\label{model_camp}
X_{ij}(t) = \lambda (g_{ij}(t)) = \lambda (\psi_j (h_i (t))) = (\lambda \circ \psi_j \circ h_i) (t),\quad t\in \I
\end{equation}
$i=1,\ldots,n$, $j=1,\ldots,p$, which is obtained after normalization of observed curves in the general latent deformation model, that is, $X_{ij} = X^{\star}_{ij} / \|X^{\star}_{ij}\|_{\infty}$, with $X^{\star}_{ij}$ generated by (\ref{model}). Let $FD$ denote a general functional satisfying the assumptions of Proposition \ref{prop_univ}.

Notice that each of the components of the process defines a univariate time-warping model in which each component's main pattern, $\gamma_j=\lambda\circ \psi_j$, acts as the latent curve. Therefore, for $j=1,\ldots, p$, $\gamma_j$ can be estimated as 
\begin{equation}\label{gam}
\hat{\gamma_j} = \hat{m}^{FD}_{X_{1:n,j}}=\arg\max_{X_{ij}} FD_{X_{1j},\ldots,X_{nj}}(X_{ij})
\end{equation}
if $\lambda$ is strictly increasing, or as
\begin{equation}\label{gam2}
\hat{\gamma_j} = \arg\max_{X_{ij}} FD_{T(X_{1j}),\ldots,T(X_{nj})}(T(X_{ij}))
\end{equation}
in the general case, according to Propositions \ref{prop_univ} and \ref{prop_univ_gen}.

Once $\gamma_j$ has been estimated, we can proceed to obtain an estimate of $h_i$ in each component of the process, therefore obtaining $p$ different estimates which can subsequently be aggregated via the sample mean as proposed by \cite{latentDef}:
\begin{equation}\label{warp_est}
\hat{h}_i = \frac{1}{p} \sum_{j=1}^p (\hat{\gamma_j}^{-1} \circ X_{ij}), \quad i=1,\ldots,n.
\end{equation}
Alternatively, and particularly suitable for large values of $p$, the component estimates of $h_i$, $\hat{h}_{ij} = \hat{\gamma_j}^{-1} \circ X_{ij}$ could be aggregated in a robust way by considering the functional median or trimmed mean.

Regarding the estimation of $\lambda$, as noted by \cite{latentDef}, we can combine all the components of the process to obtain
\begin{equation}\label{combined}
X_{ij}(t) =  (\lambda \circ \psi_j \circ h_i) (t) = \lambda \circ g_{ij},\quad t\in \I\end{equation}
where $g_{ij} = \psi_j \circ h_i \in \W_{{\I}}$ since $\psi_j\in \W_{{\I}}$ and $h_i\in \W_{{\I}}$ for all $j=1,\ldots,p$ and $i=1,\ldots,n$.
Let $\Psi$ denote a functional uniformly discrete random variable taking values in $\{\psi_1,\ldots,\psi_p\}$. Then, selecting at random one of its $p$ component curves for each individual in the sample generated by (\ref{combined}) guarantees independence and it is equivalent to observing the curves
\begin{equation}
Z_{i}(t) =  \lambda \circ g_{i},\quad t\in \I, \quad i=1,\ldots,n\end{equation}
where now $g_i$, $i=1,\ldots,n$, are i.i.d. realizations of the random variable $G = \Psi \circ H$. Under the conditions on $H$, $\Psi$ and $FD$ established in Proposition \ref{composition} below, we have $m_G^{FD} = id_{\I}$, and therefore we can consistently estimate $\lambda$ as 
\begin{equation}\label{lamb_est_1}
\hat{\lambda} = \hat{m}^{FD}_{Z_{1:n}}
\end{equation}
if $\lambda$ is strictly increasing, or as
\begin{equation}\label{lamb_est_2}
\hat{\lambda} = \arg\max_{Z_{i}} FD_{T(Z_{1}),\ldots,T(Z_{n})}(T(Z_{i}))
\end{equation}
in the general case, according to Propositions \ref{prop_univ} and \ref{prop_univ_gen}.

\begin{proposition}\label{composition}
Let $H$ be a warping process on $\W_{{\I}}$ such that for each $t\in \I$, $H(t)$ is an absolutely continuous random variable on $\I$ with central symmetry with respect to $t$, and let $\Psi$ be a warping process on $\W_{{\I}}$ such that for each $t\in \I$, $\Psi^{-1}(t)$ is centrally symmetric around $t$. 
Then, the composition $G = \Psi \circ H$ is a warping process on $\W_{{\I}}$ satisfying $m_G^{FD} = id_{\I}$, for any quantile integrated depth (QID) measure $FD$\citep{QID}.
\end{proposition}
 The proof is given in the appendix.
\begin{remark} The central symmetry assumptions on both $\Psi^{-1}$ and $H$ are necessary in Proposition \ref{composition}. Indeed, milder assumptions such as $m_H^{FD} = id_{\I}$ and $m_{\Psi^{-1}}^{FD} = id_{\I}$ or $E[\Psi^{-1}]=id_{\I}$ fail to guarantee that $m_G^{FD} = id_{\I}$ if asymmetry in either $H$ or $\Psi$ exists. \\
In particular, notice that the condition $\frac{1}{p}\sum_{j=1}^p \psi_j^{-1}(t) = t$, which is the assumption on the component-based distortion functions required in \cite{latentDef}, is not enough to guarantee the results in Proposition \ref{composition}. 
\end{remark}

\begin{remark} Notice that the QID family of functional depths\citep{QID} include all integrated and infimal functional depths as special cases. 
\end{remark}

Finally, after $\lambda$ and $\gamma_j's$ have been estimated, we estimate each of the component-based distortion functions $\psi_j$ as
\begin{equation}
\hat{\psi}_j = \hat{\lambda}^{-1} \circ \hat{\gamma_j}, \quad j=1,\ldots,p
\end{equation}

It is important to note that we do not assume any parametric representation of the individual warping functions $h_i$ or the component-based distortion functions $\psi_j$. Instead, their estimates are obtained in practice via numerical function composition and inversion of monotonic functions, relying on linear interpolation.

\subsection{Visualizing the agreement of individual warping functions}\label{WHyRA_sec}

One of the main assumptions of the latent deformation model (\ref{model}) is that the individual time distortions $h_i$, $i=1,\dots,n$ are preserved along the components of the process. That is, each individual exhibits the same time warping in every component of the process. However, this assumption might not be realistic in many situations (see Section \ref{data2} for an example). Notice that the estimator of $h_i$ given in (\ref{warp_est}) relies on that assumption, but we can still estimate $\lambda$, $\gamma_j$ and $h_{ij}$ even if it does not hold. Since in the presence of phase variation individual warping functions carry relevant information that can be used as input in classification, outlier detection, or clustering methods, among others, it is important to understand when a single warping estimate is a good indicator of each individual's time distortion across components. Therefore, we propose a diagnosis tool to assess the validity of the assumption of the equality of individual warping functions, which will allow us to decide when it is appropriate to aggregate the estimates of $h_{ij}$.

We propose a graphical tool, the Warping Hypograph Ranking Agreement (WHyRA) plot, consisting in visualizing the rank correspondence between the warping function estimates of each individual at each component of the process. Since all $\hat{h}_{ij}\,'s$ are warping functions, that is, strictly increasing functions tied at the extremes of the observation interval $\I$, a suitable way of ranking them is by means of the modified hypograph index \citep{HalfRegion11}, which accounts for the proportion of curves that lie below the graph of a given one, and therefore introduces a downwards-upwards ranking of the curves. More precisely, for a sample of curves $x_1, \ldots, x_n \in L_2(\I)$, the modified hypograph index of  $x\in \{x_1,\dots,x_n\}$ is defined as $MHI_{\{x_1,\dots,x_n\}} (x) =\frac{1}{n} \sum_{i=1}^n \ell\left(\left\{t\in {\cal I} \,\middle|\, x_i(t) \leq x(t) \right\}\right)/\ell(\I)$, where $\ell$ denotes the Lebesgue measure in $\R$.\\
The WHyRA plot is built as a matrix of scatter plots representing $MHI_{\{\hat{h}_{1j},\dots,\hat{h}_{nj}\}} (\hat{h}_{ij})$ versus $MHI_{\{\hat{h}_{1k},\dots,\hat{h}_{nk}\}} (\hat{h}_{ik})$, $i=1,\ldots,n$, for each pair $j,k=1,\ldots,p$ with $j\neq k$.  An illustration of the WHyRA plot for $p=2$ on synthetic data is given in Figure \ref{WHyRA}.  The first row corresponds to data generated under the assumption that the individual warping functions are shared across both components. The second row of Figure 2 presents data generated from a model in which this assumption is violated (see the Figure 2 caption for details). As expected, the WHyRA plot clearly distinguishes between these two scenarios, revealing markedly different patterns under the differing assumptions.\\
For larger values of $p$, the average correlation coefficient across all $(j,k)$ pairs can be used as a measure of concordance, although for a small number of dimensions, visualization of the scatter plots is preferred, since it provides insight on the ranking concordance of the warping estimates but also about possible outlying individuals, as illustrated in Figure \ref{Mat_count}.

\begin{figure}
\hspace*{-1cm}\includegraphics[width=15cm]{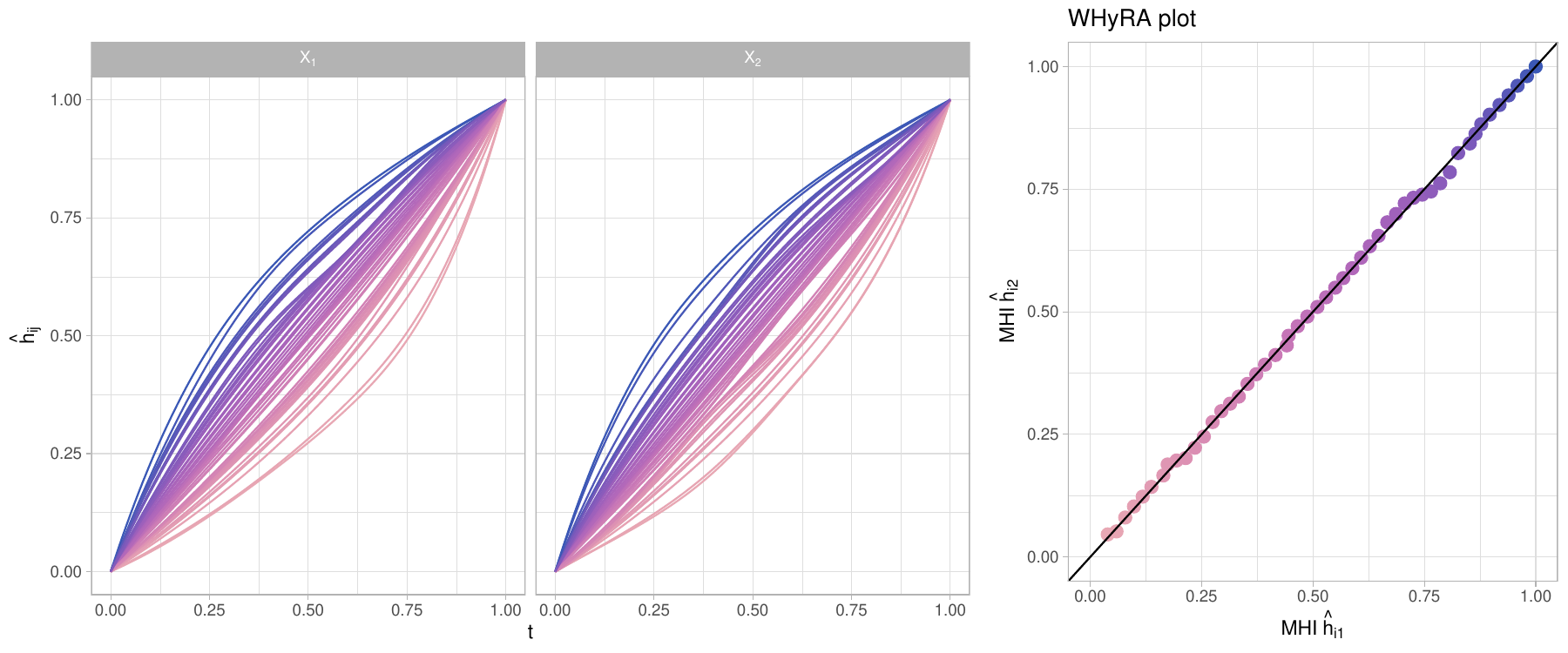}\hspace*{-1cm}\includegraphics[width=2.8cm,trim={0cm -7cm 0cm 2cm},clip]{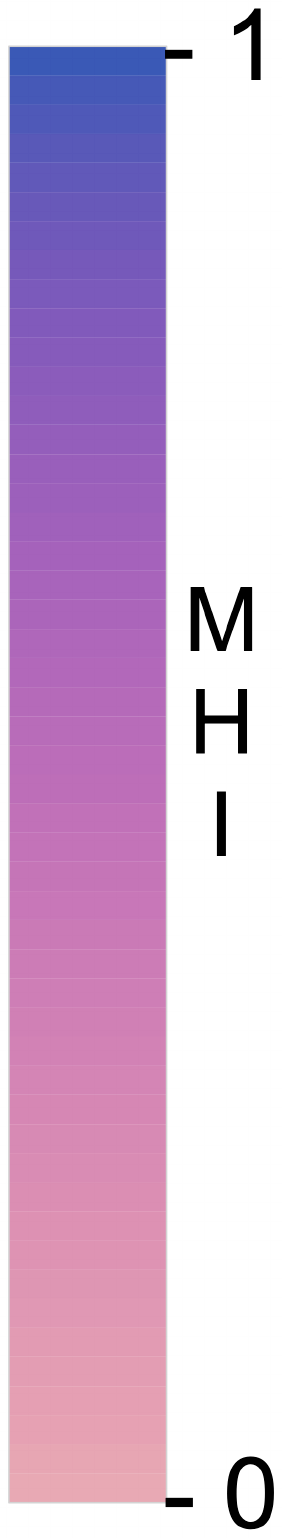}\\
\hspace*{-1cm}\includegraphics[width=15cm]{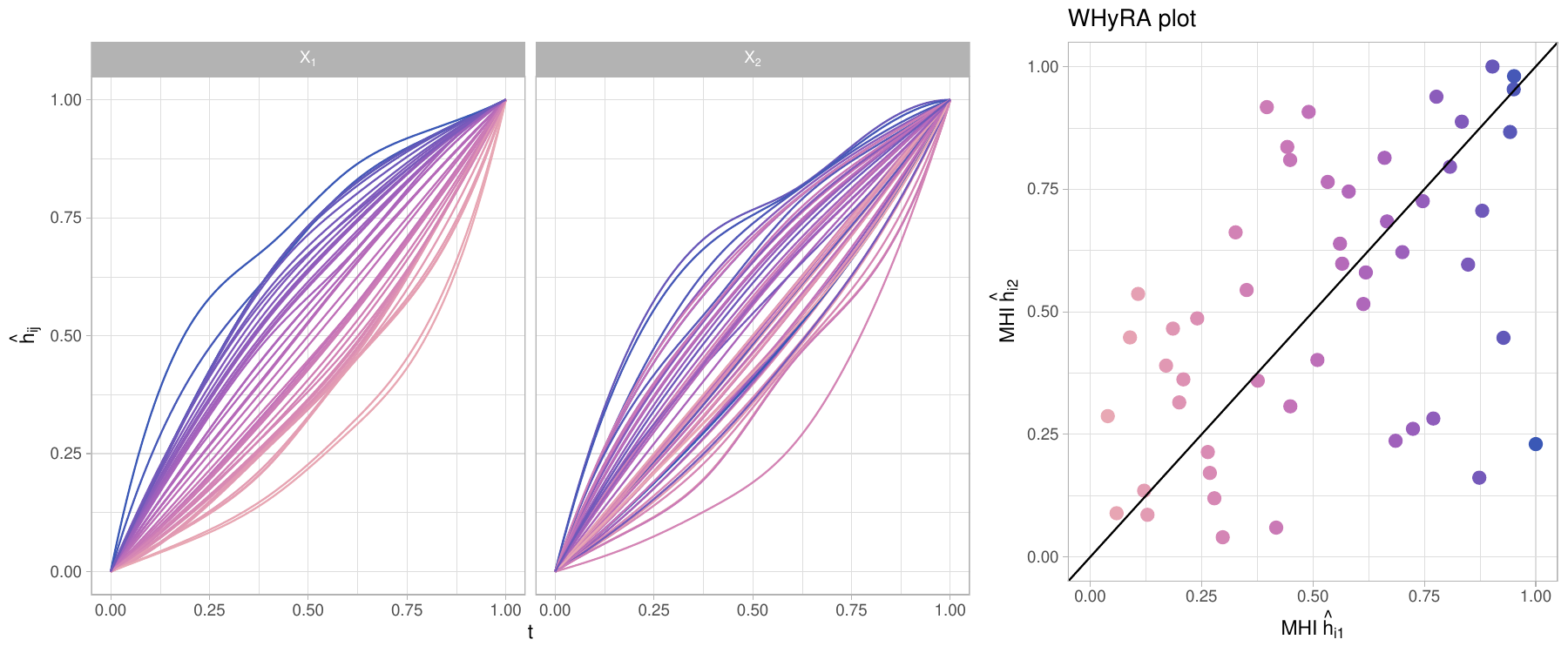}\hspace*{-1cm}\includegraphics[width=2.8cm,trim={0cm -7cm 0cm 2cm},clip]{paleta_mhi.pdf}\\
\caption{Warping function estimates and WHyRA plots from two synthetic data sets with $p=2$ and $n=50$. In each row, the first two plots present the warping function estimates $\hat{h}_{ij}$ for $j=1,2$, while the third one shows the scatter plot of the MHI values across them. Top: Data set generated according to model (\ref{model}), such that $h_i$, $i=1,\ldots,n$ are the same in the two components.  Bottom: Data set generated according to (\ref{sim_model}) with $\sigma_D=1$, so that the individual warping functions depend on $j=1,2$. Color code is defined as the modified hypograph value for each warping estimate in the first component.}\label{WHyRA}
\end{figure}

\section{Simulation study}\label{sim_sect}

In this section we analyse the behaviour of the proposed robust estimation method on synthetically generated data. In order to allow for departures from model (\ref{model}), data are generated from a more general model including both individual cross-component warping variation ($R_{ij}$) and error measurement, as proposed in \cite{latentDef}. In particular, data are generated according to model
\begin{equation}\label{sim_model}
X_{ij}(t_k) = A_{ij} (\lambda \circ \psi_j \circ h_i \circ r_{ij}) (t_k) + \varepsilon_{ijk},
\end{equation}
for $i=1,\ldots,n,\, j=1\ldots,p,\, k=1,\ldots,K$, where $t_0,\ldots,t_K$ is a collection of equispaced points in ${\cal T}=[0,1]$, with $K=101$, $n=50$ and $p=4,30$. Notice that the inclusion of the $r_{ij}$ functions breaks the assumption of shared individual warping functions across the different components of the process. Most of the elements of (\ref{sim_model}) are defined as in \cite{latentDef}, with one additional setting for each of $\Psi_j$'s and $h_i$'s, as described next:
\begin{itemize}
\item $\lambda(t) = \lambda_0(t)/||\lambda_0||_{\infty}$ with $\lambda_0(t) = 20 + 15 t^2 - 5 \cos(4 \pi t) + 3\sin(\pi t^2)$, $\,\,t \in {\cal T}$.
\item $A_{ij} \stackrel{iid}{\sim} {\cal N}(100,4)$, $\,\,i=1,\ldots,n,\,\,j=1,\ldots,p$.
\item Component-based distortion functions:
\begin{itemize}
\item $\psi$ setting 1: For $p=4$, $\psi_j(t)$, $j=1,\ldots,p$, are deterministic functions satisfying $\frac{1}{p}\sum_{j=1}^{p} \psi^{-1}_j(t) = t$. In particular, $\psi_j(t) = \frac{1}{2}Beta(t; a_j,b_j) + \frac{1}{2}t$, for $j=1,2$, and $\psi_3(t)$, $\psi_4(t)$ such that $\psi^{-1}_j(t)=2t- \psi^{-1}_{j-2}(t)$, $j=3,4$, where $Beta(t; a_j,b_j)\!\!=\!\!\dfrac{\Gamma(a_j+b_j)}{\Gamma(a_j)\Gamma(b_j)}x^{a_j-1}(1-x)^{b_j-1}$, with $a_1\!=\!2$, $b_1\!=\!2$, $a_2\!=\!2$, $b_2\!=\!\frac{1}{2}$. \\
\item $\psi$ setting 2: For $p=30$, $\psi_j(t)$, $j=1,\ldots,p$, are independent realizations of a random variable $\boldsymbol{\psi}$ satisfying $\E[ \boldsymbol{\psi}^{-1}(t) ]= t$, $\forall t$. In particular, we adopt the iterative warping generating procedure described by \cite{Loubes}. Namely, for $j=1,\ldots,p$ let $u_j\stackrel{iid}{\sim} U(10\varepsilon_W, 1-10\varepsilon_W)$ and $v_j \stackrel{iid}{\sim} U(u_i-\varepsilon_W, u_i+\varepsilon_W)$ where $\varepsilon_W$ is a warping parameter fixed to $\varepsilon_W=0.005$. Then, the component-based distortion functions are defined as $\psi^{-1}_j(t) =\varphi_i^{(M)} (t)$, where $\varphi_i^{(M)} (t)$ is obtained by iterating for $m=1,\ldots,M=2500$ the following process: 
\begin{equation}\label{wp1}
\varphi_j^{(m)}(t)=w\circ \varphi_j^{(m-1)}(t), \quad \mbox{with } w(t)=\left\{\begin{array}{ll} \dfrac{v_j}{u_j} t & \mbox{ if } 0\leq t \leq u_j \\
                                \dfrac{1-v_j}{1-u_j} t+ \dfrac{v_j-u_j}{1-u_j} & \mbox{ if } u_j < t \leq 1, \end{array}\right.
\end{equation} 
\end{itemize}

\item Individual warping functions: 
\begin{itemize}
\item $h$ setting 1: $$h_{i}^{-1}(t) = \dfrac{\exp (t w_{i})-1}{\exp (w_{i})-1}, \quad w_{i} \stackrel{iid}{\sim} {\cal N}(0,\sigma_W^2), \quad i=1,\ldots,n.$$
\item $h$ setting 2: $\quad h_{i}^{-1}(t) = \varphi_i^{(M)}(t)$, $\quad i=1,\ldots,n$,\vspace{0.3cm}\\with $\varphi_i^M$ generated as in (\ref{wp1}) ($M=2500$ and varying values of $\varepsilon_W$).
\end{itemize}
\item $r_{ij}^{-1}(t) = \dfrac{\exp (t d_{ij})-1}{\exp (d_{ij})-1}$, where $d_{ij} \stackrel{iid}{\sim} {\cal N}(0,\sigma_D^2)$,  $\,\,i=1,\ldots,n,\,\,j=1,\ldots,p$.
\item $\varepsilon_{ijk}\stackrel{iid}{\sim} {\cal N}(0,\sigma_E^2)$ $\,\,i=1,\ldots,n,\,\,j=1,\ldots,p,\,\,k=1,\ldots,K$.

\end{itemize}
The parameter values are set to $\sigma_W= 0.5, 1$, $\varepsilon_W= 0.005, 0.0075$, $\sigma_D=0,  0.5, 1$,  $\sigma_E= 0, 1, 5, 10$. Graphical representations of the component distortion functions $\psi_j$ and the individual warping functions $h_i$ are presented in Figure \ref{figs_sim}. Notice that the warping functions under setting 1 are always above or below the diagonal whereas those under setting 2 might cross the diagonal of $\I$. In general, for the chosen parameter value, setting 2 for $h_i$ represents a higher amount of individual warping. For both settings 1 and 2, the generated individual warping functions are i.i.d. realizations of some random variable $H$, satisfying $\E[ H^{-1} ]= id_{\I}$ and $m_H^{MBD}=id_{\I}$, where $MBD$ denotes the modified band depth. However, notice that the random variables $H$ defined by settings 1 and 2 do not have centrally symmetric marginals around $t$ for all values of $t \in \I$ (see Figure \ref{figs_sim}). The same happens for the component-based distortion functions used, so the conditions to ensure the result of Proposition \ref{composition} are not satisfied. Some simulated data sets for different settings and parameter values are presented in Figure \ref{figs_simdata}.

\begin{figure}
\begin{center}
\includegraphics[width=14cm]{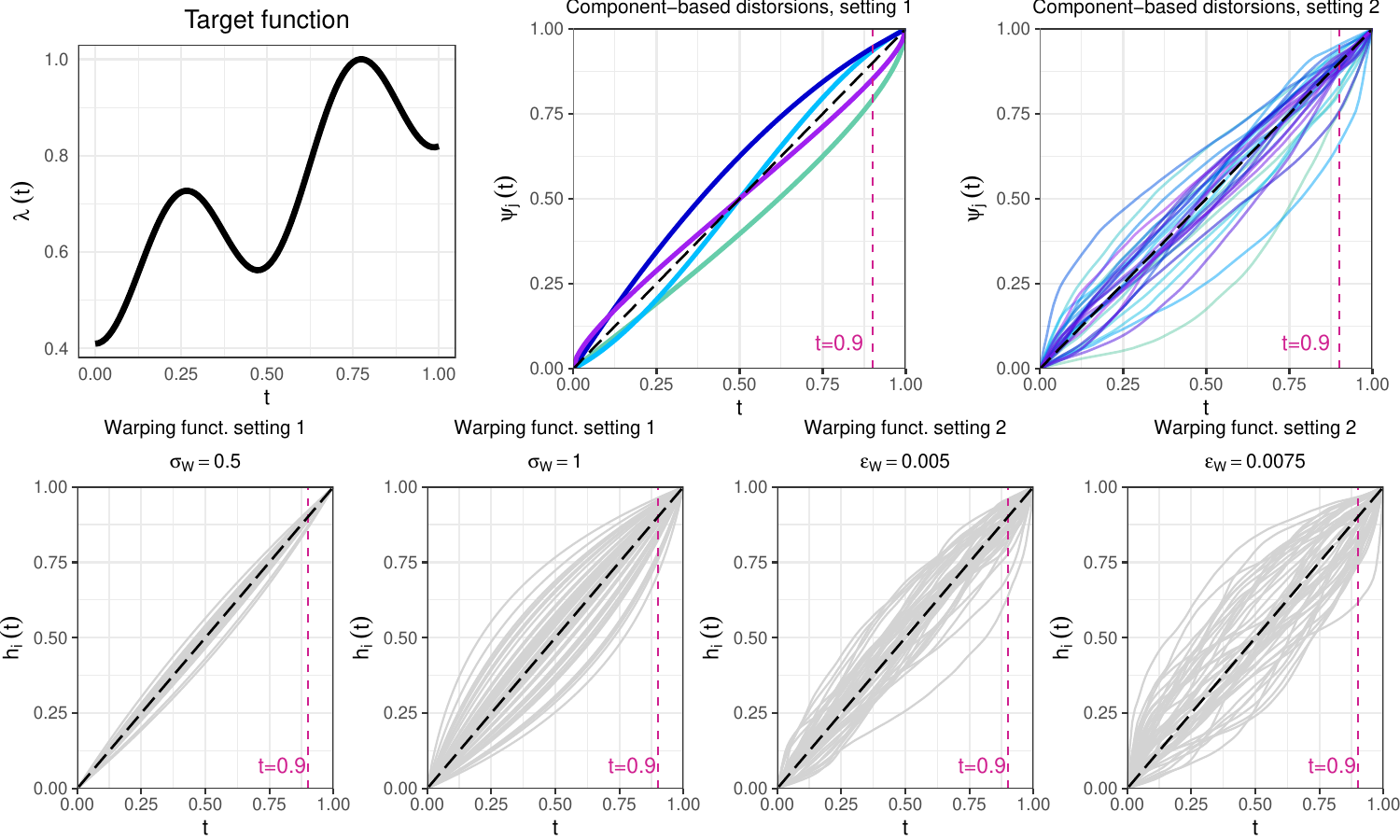}
\end{center}
\caption{Illustration of the elements of the simulation study. Top: target function (left) and component-based distortion functions ($\psi_j$'s) for setting 1 (middle, $p=4$) and setting 2 (right, $p=30$). Bottom: individual warping functions ($h_i$'s) for settings 1 and 2 and different parameter values. In all simulated component-based distortions and individual warping functions a transversal section at $t=0.9$ is presented to help visualize that the resulting univariate distribution is not symmetric around $t=0.9$. The black dashed line represents the diagonal of the interval $\I=[0,1]$.}\label{figs_sim}
\end{figure}

\begin{figure}
\begin{center}
\hspace*{-1cm}\includegraphics[width=15.5cm]{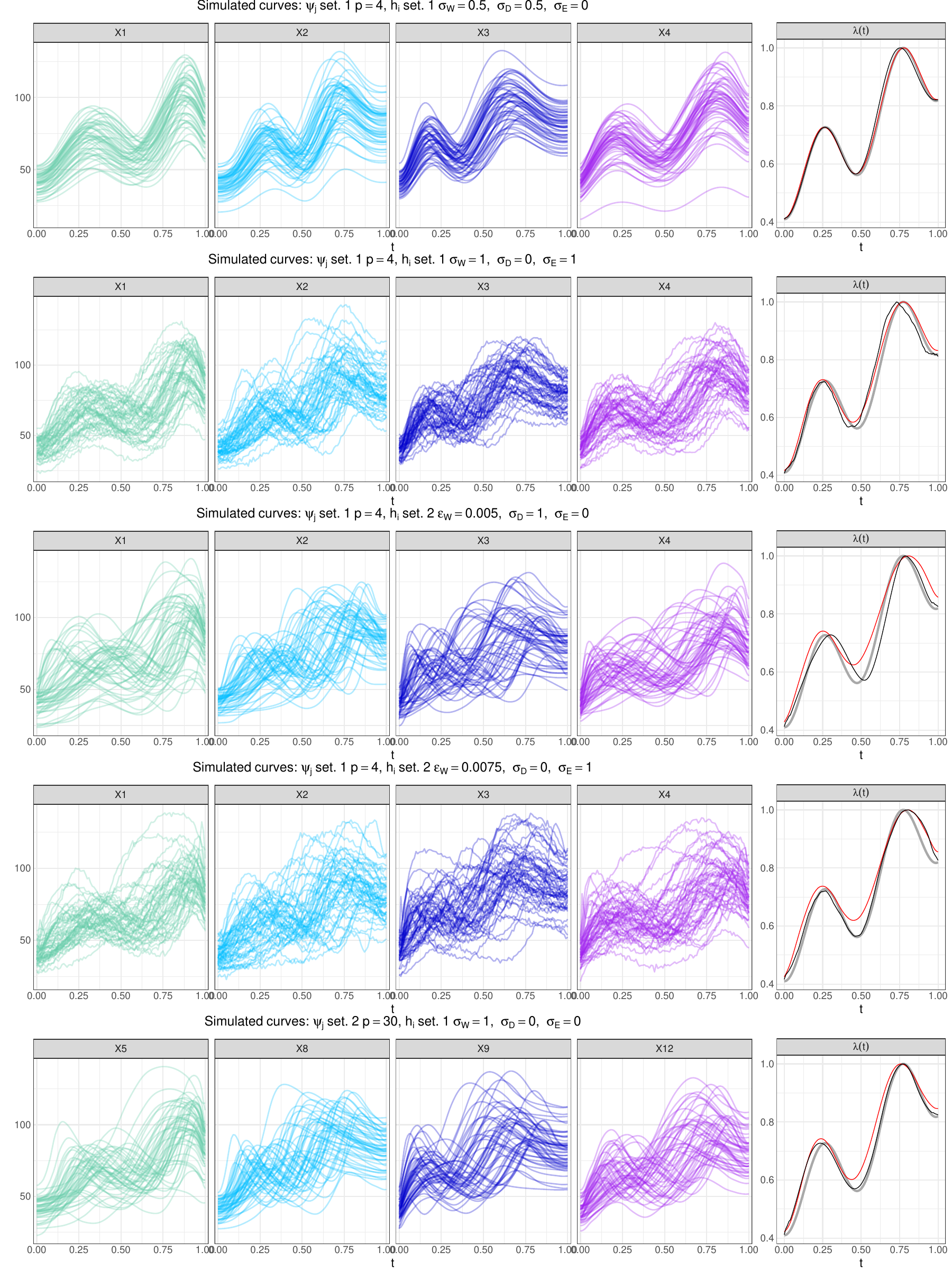}\end{center}
\caption{Five simulated data sets under model (\ref{sim_model}). For the last one, the number of components is $p=30$ and only four randomly selected components are shown. The fifth plot on each row represents the target function, $\lambda(t)$, as a thick gray solid line, together with the depth-based and the CM estimates (black and red lines, respectively).}\label{figs_simdata}
\end{figure}

For each simulation setting and simulation run, estimates of $\lambda$, $h_i$, $\gamma_j$ are obtained as described in Section \ref{estimators}, with the choice of the modified band depth as functional depth measure. For the estimator of $\lambda$, considering the pooled sample of all the observed curves across the $p$ components, rather than a single randomly selected representative for each individual, provides better small sample results, which are the ones shown here. We also consider $\hat{X}_{ij} = \hat{\gamma}_j \circ \hat{h}_i$ as an estimator of individual observations. Then, the following integrated error measures are computed:
 $$ LISE = \int_I (\hat{\lambda}(t)-\lambda(t))^2 dt\quad   \quad \quad HMISE =\dfrac{1}{n} \sum_{i=1}^n \int_I (\hat{h}_i(t)-h_i(t))^2 dt$$
 $$XMISE =\dfrac{1}{np} \sum_{i=1}^n\sum_{j=1}^p \int_I (\hat{X}_{ij}(t)-X_{ij}(t))^2 dt$$
Their mean values and standard deviations over $N=100$ simulation runs are presented in Tables \ref{res1}-\ref{res4} compared to those obtained with the estimation method of \cite{latentDef}, which will be referred to as CM method in the following.  The method works by first providing estimates $\hat{h}_{ij}$ of the individual warping functions by solving $p$ registration problems in each of the components of the multivariate sample, and then averaging them to obtain $\hat{h}_i$, $i=1,\ldots,n$, and $\hat{\gamma}_j = \frac{1}{n} \sum_{i=1}^n (X_{ij} \circ \hat{h}^{-1}_{i})$, $j=1\ldots,p$.  Notice that this corresponds to an inversion in the order of the steps to estimate the component target functions and individual warping functions with respect to our procedure (equations (\ref{gam}-\ref{gam2}) and (\ref{warp_est})).  Finally, $\lambda$ is obtained as the mean of the aligned global sample, after solving a registration problem on the pool sample of all the components. The method has been implemented here by using the \texttt{R} functions provided by the authors in the Supplementary Material of their article.

In view of the results, we can see how LISE is significantly lower for the CM method in the cases of slight to mild time warping. However, as warping variability ($\sigma_D$, $\varepsilon_W$) and/or additional distortion ($\sigma_D$) increases, our depth-based method presents better results. An illustration of the two $\lambda$ estimates in different simulation scenarios is presented in Figure \ref{figs_simdata}. This behaviour is consistent across different generation mechanisms for $\psi$ (settings 1 and 2) and different sample sizes ($n=50$ or $100$). Moreover, as the sample size increases, the reduction in LISE is noticeable larger for our depth-based method than for the CM estimator.

As for the integrated square error for the warping functions and the observed curves, HMISE and XMISE, our method provides consistently better results except for retrieving the observed curves in the presence of additive noise. In fact, it is noticeable how our $\lambda$ and $X_{ij}$ estimates are more sensitive to the presence of additive noise than the mean-based estimates, as expected. And this is despite incorporating a moving average smoothing pre-step in our estimation strategy. Larger values of $\sigma_E$ were tested that corroborate this (results not shown). 

Regarding the computation time, our method is significantly faster since it does not involve registration of the curves. In particular, the mean (standard deviation) times in seconds over 100 simulation runs are 0.17 (0.09), 0.32 (0.05), 1.12 (0.15) and 2.19 (0.22) for our procedure for $p=4$ and $n=50, 100$ and $p=30$ and $n=50, 100$,  respectively, and 11.02 (0.49),  44.55 (0.76),  85.39 (2.21) and 349.60 (8.60) for the CM method for the same sample sizes (\texttt{R} version 4.5.1 running on an Apple M1 8-core CPU with 16Gb of memory).

The specific details of the implementation used in this simulation study can be found in the R code available at \url{https://github.com/aarribasuc3m/DepthBased_Multivariate_FunctionalTimeWarping}.

\subsection{Simulations under outlier contamination}

In order to evaluate the robustness of our method to the presence of outliers, we now present an extension of the previous simulation study in which a proportion $c$ of the $n$ observations are generated from (\ref{sim_model}) with a different common amplitude function $\lambda_c(t)$ defined as
\begin{equation}\label{lamb_cont}
\lambda_c(t) = \lambda^c_0(t)/||\lambda^c_0||_{\infty}\,\mbox{ with }\,\lambda^c_0(t) = t + \exp\{-25(t-0.5)^2\}, \,\,t \in {\cal T}=[0,1].
\end{equation}
The proportion of contaminated observations in each simulation run is given by $c=0.05$ or $0.1$. The rest of the simulation settings are the same as described above. 

The results are summarized in Tables \ref{res_outliers} and \ref{res_outliers2}, and illustrated in Figure \ref{figs_simdata_outs} in the Appendix. For the sake of conciseness, and since different specifications seem to provide similar outcomes as discussed above, we only present results obtained with $n=50$, $\psi$ generated under setting 1 ($p=4$), and $\sigma_E = 0$. We can observe how the LISE remains almost unchanged across different contamination levels and the error in the estimation of the warping functions increases only slightly for our depth-based approach, whereas the performance of the CM method is more importantly affected by the presence of the outlying observations.

\begin{table}
$\mathbf{\psi}$ {\bf setting 1.} $\quad\mathbf{n=50}, \quad\mathbf{p=4}$\vspace{0.3cm}\\
\tiny{
\hspace{-0.65cm}\begin{tabular}{cc|cc|cc|cc}
    \bf{ $\mathbf{h_i}$ setting 1}&&\multicolumn{6}{l}{\bf{ Additive noise: $\sigma_E=0$}}\\
    \hline
  Warping & Nuisance  &\multicolumn{2}{l|}{LISE $\times 10^3$}&\multicolumn{2}{l|}{HMISE$\times 10^3$}&\multicolumn{2}{l}{XMISE}\\
   distortion & distortion & $\hat{\lambda}$&$\hat{\lambda}_{CM}$& $\hat{h}_i$&$\hat{h}_{i,CM}$& $\hat{X}_{ij}$&$\hat{X}_{ij,CM}$\\
    \hline
    &$\sigma_D\!=\!0$&2.04  ( 0.63 )&\bf 0.06  ( 0.04 )&\bf 0.02  ( 0.02 )&0.05  ( 0.04 )&\bf 0.11  ( 0.03 )&0.28  ( 0.07 )\\
    \cline{2-8} 
    $\sigma_W\!=\!0.5$&$\sigma_D\!=\!0.5$&0.91  ( 0.8 )&\bf 0.08  ( 0.05 )&\bf 0.13  ( 0.03 )&0.21  ( 0.06 )&\bf 6.7  ( 0.79 )&6.97  ( 0.78 )\\
    \cline{2-8} 
    &$\sigma_D\!=\!1$&\bf 0.32  ( 0.4 )&0.99  ( 0.47 )&1.73  ( 0.31 )&\bf 1.54  ( 0.26 )&\bf 74.97  ( 7.25 )&78.77  ( 8.48 )\\
     \hline
    &$\sigma_D\!=\!0$&\bf 0.5  ( 0.88 )&1.15  ( 0.94 )&\bf 0.39  ( 0.47 )&1.22  ( 0.71 )&\bf 1.65  ( 1.03 )&16.59  ( 8.54 )\\
     \cline{2-8} 
   $\phantom{iii}\sigma_W\!=\!1\phantom{iii}$&$\sigma_D\!=\!0.5$&\bf 0.38  ( 0.42 )&1.16  ( 0.8 )&\bf 0.4  ( 0.32 )&1.28  ( 0.61 )&\bf 8.02  ( 1.19 )&23.41  ( 8.63 )\\
    \cline{2-8} 
   &$\sigma_D\!=\!1$&\bf 0.86  ( 0.97 )&3.17  ( 1.1 )&\bf 1.84  ( 0.49 )&2.84  ( 0.66 )&\bf 71.67  ( 6.25 )&100.16  ( 11.93 )\\
    \hline
    \end{tabular}}\vspace{0.3cm}\\
\tiny{
\begin{tabular}{cc|cc|cc|cc}
\bf{ $\mathbf{h_i}$ setting 1}&&\multicolumn{6}{l}{\bf{ Additive noise: $\sigma_E=1$}}\\    \hline
  Warping & Nuisance  &\multicolumn{2}{l|}{LISE $\times 10^3$}&\multicolumn{2}{l|}{HMISE$\times 10^3$}&\multicolumn{2}{l}{XMISE}\\
   distortion & distortion & $\hat{\lambda}$& $\hat{\lambda}_{CM}$& $\hat{h}_i$&$\hat{h}_{i,CM}$& $\hat{X}_{ij}$&$\hat{X}_{ij,CM}$\\
    \hline
    &$\sigma_D\!=\!0$&2.32  ( 0.89 )&\bf 0.06  ( 0.03 )&0.08  ( 0.03 )&\bf 0.04  ( 0.03 )&2.51  ( 0.39 )&\bf 1.44  ( 0.08 )\\
     \cline{2-8}
    $\sigma_W\!=\!0.5$&$\sigma_D\!=\!0.5$&1.78  ( 1.09 )&\bf 0.07  ( 0.04 )&0.19  ( 0.05 )&0.19  ( 0.05 )&8.93  ( 0.91 )&\bf 7.94  ( 0.8 )\\
     \cline{2-8}
    &$\sigma_D\!=\!1$&1.01  ( 0.92 )&\bf 0.87  ( 0.39 )&1.65  ( 0.28 )&\bf 1.49  ( 0.27 )&\bf 76.48  ( 6.48 )&78.16  ( 7.22 )\\
     \hline
    &$\sigma_D\!=\!0$&1.33  ( 1.16 )&\bf 0.98  ( 0.67 )&\bf 0.39  ( 0.29 )&1.17  ( 0.6 )&\bf 4.78  ( 1.13 )&17.73  ( 7.55 )\\
     \cline{2-8}
$\phantom{iii}\sigma_W\!=\!1\phantom{iii}$&$\sigma_D\!=\!0.5$&1.23  ( 1 )&\bf 1.14  ( 0.77 )&\bf 0.5  ( 0.37 )&1.36  ( 0.58 )&\bf 11  ( 1.47 )&25.82  ( 7.65 )\\
    \cline{2-8}
   &$\sigma_D\!=\!1$&\bf 1.91  ( 1.41 )&3.03  ( 1.07 )&\bf 1.79  ( 0.38 )&2.86  ( 0.59 )&\bf 75.99  ( 6.48 )&102.25  ( 11.82 )\\
    \hline
    \end{tabular}}\vspace{0.3cm}\\
    \tiny{
\begin{tabular}{cc|cc|cc|cc}
\bf{ $\mathbf{h_i}$ setting 2}&&\multicolumn{6}{l}{\bf{ Additive noise: $\sigma_E=0$}}\\    \hline
  Warping & Nuisance  &\multicolumn{2}{l|}{LISE $\times 10^3$}&\multicolumn{2}{l|}{HMISE$\times 10^3$}&\multicolumn{2}{l}{XMISE}\\
   distortion & distortion & $\hat{\lambda}$& $\hat{\lambda}_{CM}$& $\hat{h}_i$&$\hat{h}_{i,CM}$& $\hat{X}_{ij}$&$\hat{X}_{ij,CM}$\\
    \hline
    &$\sigma_D\!=\!0$&0.85  ( 0.65 )&\bf 0.82  ( 0.47 )&\bf 0.41  ( 0.28 )&0.81  ( 0.4 )&\bf 2.89  ( 2.21 )&10.93  ( 5.69 )\\
     \cline{2-8}
    $\varepsilon_W\!=\!0.005$&$\sigma_D\!=\!0.5$&\bf 0.84  ( 0.56 )&1.01  ( 0.62 )&\bf 0.49  ( 0.28 )&1.02  ( 0.45 )&\bf 10.28  ( 1.93 )&17.59  ( 4.45 )\\
     \cline{2-8}
    &$\sigma_D\!=\!1$&\bf 1.38  ( 0.87 )&2.52  ( 1.03 )&\bf 1.96  ( 0.46 )&2.43  ( 0.55 )&\bf 72.13  ( 5.71 )&88.62  ( 9.25 )\\
     \hline
    &$\sigma_D\!=\!0$&\bf 1.93  ( 1.04 )&3.28  ( 1.41 )&\bf 1.14  ( 0.58 )&4.13  ( 1.5 )&\bf 9.87  ( 5.36 )&53.25  ( 19.87 )\\
     \cline{2-8}
   $\varepsilon_W\!=\!0.0075$&$\sigma_D\!=\!0.5$&\bf 1.91  ( 1.01 )&3.39  ( 1.38 )&\bf 1.19  ( 0.57 )&4.35  ( 1.36 )&\bf 16.27  ( 3.68 )&60.3  ( 18.13 )\\
    \cline{2-8}
   &$\sigma_D\!=\!1$&\bf 2.33  ( 1.44 )&5.87  ( 1.79 )&\bf 2.61  ( 0.6 )&5.83  ( 1.4 )&\bf 73.26  ( 7.98 )&118.46  ( 17.85 )\\
    \hline
    \end{tabular}}\vspace{0.3cm}\\
    \tiny{
 \begin{tabular}{cc|cc|cc|cc}
   \bf{ $\mathbf{h_i}$ setting 2}&&\multicolumn{6}{l}{\bf{ Additive noise: $\sigma_E=1$}}\\    \hline
  Warping & Nuisance  &\multicolumn{2}{l|}{LISE $\times 10^3$}&\multicolumn{2}{l|}{HMISE$\times 10^3$}&\multicolumn{2}{l}{XMISE}\\
   distortion & distortion & $\hat{\lambda}$& $\hat{\lambda}_{CM}$& $\hat{h}_i$&$\hat{h}_{i,CM}$& $\hat{X}_{ij}$&$\hat{X}_{ij,CM}$\\
    \hline   
    &$\sigma_D\!=\!0$&1.16  ( 0.89 )&\bf 0.76  ( 0.37 )&\bf 0.45  ( 0.27 )&0.8  ( 0.33 )&\bf 7.55  ( 2.77 )&12.07  ( 4.49 )\\
     \cline{2-8}
    $\varepsilon_W\!=\!0.005$&$\sigma_D\!=\!0.5$&1.06  ( 0.68 )&\bf 0.86  ( 0.45 )&\bf 0.54  ( 0.28 )&0.99  ( 0.44 )&\bf 13.98  ( 2.51 )&19.56  ( 6.21 )\\
     \cline{2-8}
    &$\sigma_D\!=\!1$&\bf 1.53  ( 0.98 )&2.56  ( 1.03 )&\bf 1.92  ( 0.44 )&2.35  ( 0.52 )&\bf 76.94  ( 8.05 )&91.28  ( 11.37 )\\
     \hline
   &$\sigma_D\!=\!0$&\bf 2.17  ( 1.48 )&3.02  ( 1.36 )&\bf 1.18  ( 0.64 )&4.07  ( 1.48 )&\bf 13.76  ( 4.86 )&52.55  ( 18.64 )\\
     \cline{2-8}
   $\varepsilon_W\!=\!0.0075$&$\sigma_D\!=\!0.5$&\bf 2.43  ( 1.55 )&3.37  ( 1.51 )&\bf 1.28  ( 0.65 )&4.37  ( 1.63 )&\bf 21.17  ( 5.06 )&59.21  ( 21.2 )\\
    \cline{2-8}
   &$\sigma_D\!=\!1$&\bf 2.67  ( 1.48 )&5.55  ( 1.81 )&\bf 2.6  ( 0.78 )&5.81  ( 1.38 )&\bf 78.54  ( 8.23 )&119.56  ( 18.83 )\\
    \hline
    \end{tabular}}
    \caption{LISE, HMISE and XMISE means and standard deviations of depth-based estaimators ($\hat{\lambda}$, $\hat{h}_i$, $\hat{X}_{ij}$) and those proposed in \cite{latentDef} ($\hat{\lambda}_{CM}$, $\hat{h}_{i,CM}$, $\hat{X}_{ij,CM}$) over $N=100$ simulation runs for $p=4$ deterministic component distortion functions ($\psi$ setting 1) and $n=50$. The lowest integrated error values in each setting are highlighted in bold.}\label{res1}
\end{table}

\begin{table}
$\mathbf{\psi}$ {\bf setting 1.} $\quad\mathbf{n=100}, \quad\mathbf{p=4}$\vspace{0.3cm}\\
\tiny{
\begin{tabular}{cc|cc|cc|cc}
     \bf{ $\mathbf{h_i}$ setting 1}&&\multicolumn{6}{l}{\bf{ Additive noise: $\sigma_E=0$}}\\
    \hline
  Warping & Nuisance  &\multicolumn{2}{l|}{LISE $\times 10^3$}&\multicolumn{2}{l|}{HMISE$\times 10^3$}&\multicolumn{2}{l}{XMISE}\\
   distortion & distortion & $\hat{\lambda}$& $\hat{\lambda}_{CM}$& $\hat{h}_i$&$\hat{h}_{i,CM}$& $\hat{X}_{ij}$&$\hat{X}_{ij,CM}$\\
    \hline
    &$\sigma_D\!=\!0$&1.87  ( 0.68 )&\bf{ 0.05  ( 0.02 )}&\bf{ 0.01  ( 0.01 )}&0.03  ( 0.03 )&\bf{ 0.11  ( 0.02 )}&0.16  ( 0.02 )\\
    \cline{2-8} 
    $\sigma_W\!=\!0.5$&$\sigma_D\!=\!0.5$&0.54  ( 0.45 )&\bf{ 0.06  ( 0.02 )}&\bf{ 0.12  ( 0.02 )}&0.18  ( 0.03 )&\bf{ 6.6  ( 0.57 )}&6.77  ( 0.57 )\\
    \cline{2-8} 
    &$\sigma_D\!=\!1$&\bf{ 0.16  ( 0.2 )}&0.75  ( 0.27 )&1.67  ( 0.24 )&\bf{ 1.38  ( 0.18 )}&\bf{ 74.8  ( 5.11 )}&78.4  ( 5.56 )\\
     \hline
    &$\sigma_D\!=\!0$&\bf{ 0.19  ( 0.16 )}&0.76  ( 0.41 )&\bf{ 0.2  ( 0.16 )}&1.03  ( 0.41 )&\bf{ 1.63  ( 0.55 )}&16.93  ( 5.82 )\\
     \cline{2-8} 
   $\phantom{iii}\sigma_W\!=\!1\phantom{iii}$&$\sigma_D\!=\!0.5$&\bf{ 0.23  ( 0.29 )}&0.95  ( 0.52 )&\bf{ 0.29  ( 0.16 )}&1.18  ( 0.39 )&\bf{ 7.63  ( 0.75 )}&23.84  ( 5.81 )\\
    \cline{2-8} 
   &$\sigma_D\!=\!1$&\bf{ 0.47  ( 0.48 )}&2.93  ( 0.79 )&\bf{ 1.67  ( 0.24 )}&2.78  ( 0.5 )&\bf{ 71.85  ( 4.9 )}&102.99  ( 9.24 )\\
    \hline
    \end{tabular}}\vspace{0.3cm}\\
\tiny{
\begin{tabular}{cc|cc|cc|cc}
    \bf{ $\mathbf{h_i}$ setting 1}&&\multicolumn{6}{l}{\bf{ Additive noise: $\sigma_E=1$}}\\    \hline
  Warping & Nuisance  &\multicolumn{2}{l|}{LISE $\times 10^3$}&\multicolumn{2}{l|}{HMISE$\times 10^3$}&\multicolumn{2}{l}{XMISE}\\
   distortion & distortion & $\hat{\lambda}$& $\hat{\lambda}_{CM}$& $\hat{h}_i$&$\hat{h}_{i,CM}$& $\hat{X}_{ij}$&$\hat{X}_{ij,CM}$\\
    \hline
    &$\sigma_D\!=\!0$&2.22  ( 0.88 )&\bf{ 0.04  ( 0.01 )}&0.07  ( 0.02 )&\bf{ 0.02  ( 0.02 )}&2.41  ( 0.27 )&\bf{ 1.31  ( 0.03 )}\\
     \cline{2-8}
    $\sigma_W\!=\!0.5$&$\sigma_D\!=\!0.5$&1.31  ( 0.95 )&\bf{ 0.05  ( 0.02 )}&0.18  ( 0.02 )&\bf{ 0.17  ( 0.03 )}&8.87  ( 0.65 )&\bf{ 7.98  ( 0.51 )}\\
     \cline{2-8}
    &$\sigma_D\!=\!1$&0.92  ( 0.78 )&\bf{ 0.71  ( 0.23 )}&1.62  ( 0.22 )&\bf{ 1.34  ( 0.17 )}&\bf{ 76.82  ( 5.24 )}&79.61  ( 5.46 )\\
     \hline
    &$\sigma_D\!=\!0$&0.9  ( 0.86 )&\bf{ 0.77  ( 0.46 )}&\bf{ 0.32  ( 0.25 )}&1  ( 0.38 )&\bf{ 4.54  ( 0.88 )}&17.62  ( 5.73 )\\
     \cline{2-8}
   $\phantom{iii}\sigma_W\!=\!1\phantom{iii}$&$\sigma_D\!=\!0.5$&0.85  ( 0.89 )&\bf{ 0.74  ( 0.43 )}&\bf{ 0.39  ( 0.18 )}&1.07  ( 0.38 )&\bf{ 10.48  ( 1.03 )}&23.77  ( 5.59 )\\
    \cline{2-8}
   &$\sigma_D\!=\!1$&\bf{ 1.37  ( 0.87 )}&2.97  ( 0.86 )&\bf{ 1.72  ( 0.25 )}&2.85  ( 0.43 )&\bf{ 74  ( 4.65 )}&103.88  ( 8.79 )\\
    \hline
    \end{tabular}}\vspace{0.3cm}\\
    \tiny{
\begin{tabular}{cc|cc|cc|cc}
   \bf{ $\mathbf{h_i}$ setting 2}&&\multicolumn{6}{l}{\bf{ Additive noise: $\sigma_E=0$}}\\    \hline
  Warping & Nuisance  &\multicolumn{2}{l|}{LISE $\times 10^3$}&\multicolumn{2}{l|}{HMISE$\times 10^3$}&\multicolumn{2}{l}{XMISE}\\
   distortion & distortion & $\hat{\lambda}$& $\hat{\lambda}_{CM}$& $\hat{h}_i$&$\hat{h}_{i,CM}$& $\hat{X}_{ij}$&$\hat{X}_{ij,CM}$\\
    \hline
 &$\sigma_D\!=\!0$&\bf{ 0.62  ( 0.35 )}&0.63  ( 0.3 )&\bf{ 0.28  ( 0.14 )}&0.68  ( 0.25 )&\bf{ 2.55  ( 1.27 )}&10.31  ( 3.84 )\\
     \cline{2-8}
    $\varepsilon_W\!=\!0.005$&$\sigma_D\!=\!0.5$&\bf{ 0.65  ( 0.36 )}&0.69  ( 0.29 )&\bf{ 0.36  ( 0.13 )}&0.81  ( 0.3 )&\bf{ 9.63  ( 1.44 )}&16.94  ( 4.72 )\\
     \cline{2-8}
    &$\sigma_D\!=\!1$&\bf{ 1.05  ( 0.61 )}&2.39  ( 0.7 )&\bf{ 1.84  ( 0.26 )}&2.31  ( 0.33 )&\bf{ 72.05  ( 4.8 )}&90.94  ( 8.99 )\\
     \hline
    &$\sigma_D\!=\!0$&\bf{ 1.55  ( 0.69 )}&3.02  ( 0.93 )&\bf{ 0.9  ( 0.37 )}&4.04  ( 1.11 )&\bf{ 9.01  ( 3.78 )}&54.27  ( 15.96 )\\
     \cline{2-8}
   $\varepsilon_W\!=\!0.0075$&$\sigma_D\!=\!0.5$&\bf{ 1.57  ( 0.74 )}&2.91  ( 0.83 )&\bf{ 0.95  ( 0.35 )}&3.96  ( 0.94 )&\bf{ 15.47  ( 2.58 )}&57.05  ( 11.81 )\\
    \cline{2-8}
   &$\sigma_D\!=\!1$&\bf{ 1.78  ( 0.84 )}&5.38  ( 1.39 )&\bf{ 2.34  ( 0.4 )}&5.78  ( 1.18 )&\bf{ 72.09  ( 4.68 )}&119.99  ( 13.48 )\\
    \hline
    \end{tabular}}\vspace{0.3cm}\\
    \tiny{
\begin{tabular}{cc|cc|cc|cc}
   \bf{ $\mathbf{h_i}$ setting 2}&&\multicolumn{6}{l}{\bf{ Additive noise: $\sigma_E=1$}}\\    \hline
  Warping & Nuisance  &\multicolumn{2}{l|}{LISE $\times 10^3$}&\multicolumn{2}{l|}{HMISE$\times 10^3$}&\multicolumn{2}{l}{XMISE}\\
   distortion & distortion & $\hat{\lambda}$& $\hat{\lambda}_{CM}$& $\hat{h}_i$&$\hat{h}_{i,CM}$& $\hat{X}_{ij}$&$\hat{X}_{ij,CM}$\\
    \hline   
   &$\sigma_D\!=\!0$&0.85  ( 0.48 )&\bf{ 0.6  ( 0.24 )}&\bf{ 0.35  ( 0.14 )}&0.69  ( 0.25 )&\bf{ 6.67  ( 1.83 )}&11.48  ( 3.59 )\\
     \cline{2-8}
    $\varepsilon_W\!=\!0.005$&$\sigma_D\!=\!0.5$&0.98  ( 0.65 )&\bf{ 0.71  ( 0.35 )}&\bf{ 0.49  ( 0.22 )}&0.84  ( 0.31 )&\bf{ 13.5  ( 1.82 )}&17.89  ( 3.7 )\\
     \cline{2-8}
    &$\sigma_D\!=\!1$&\bf{ 1.24  ( 0.77 )}&2.31  ( 0.71 )&\bf{ 1.79  ( 0.24 )}&2.3  ( 0.31 )&\bf{ 76.05  ( 6.04 )}&92.39  ( 8.74 )\\
     \hline
   &$\sigma_D\!=\!0$&\bf{ 1.87  ( 0.87 )}&3  ( 0.95 )&\bf{ 1.06  ( 0.39 )}&4.11  ( 1.09 )&\bf{ 14.3  ( 3.8 )}&56.48  ( 13.78 )\\
     \cline{2-8}
   $\varepsilon_W\!=\!0.0075$&$\sigma_D\!=\!0.5$&\bf{ 1.84  ( 1.02 )}&3.24  ( 1.13 )&\bf{ 1.12  ( 0.36 )}&4.48  ( 1.17 )&\bf{ 20.11  ( 3.15 )}&65.07  ( 15.11 )\\
    \cline{2-8}
   &$\sigma_D\!=\!1$&\bf{ 2.15  ( 1.07 )}&5.17  ( 1.2 )&\bf{ 2.43  ( 0.44 )}&5.89  ( 1.1 )&\bf{ 76.24  ( 5.98 )}&120.55  ( 12.97 )\\
    \hline
    \end{tabular}}
    \caption{LISE, HMISE and XMISE means and standard deviations of depth-based estaimators ($\hat{\lambda}$, $\hat{h}_i$, $\hat{X}_{ij}$) and those proposed in \cite{latentDef} ($\hat{\lambda}_{CM}$, $\hat{h}_{i,CM}$, $\hat{X}_{ij,CM}$) over $N=100$ simulation runs for $p=4$ deterministic component distortion functions ($\psi$ setting 1) and $n=100$. The lowest integrated error values in each setting are highlighted in bold.}\label{res2}
\end{table}

\begin{table}
$\mathbf{\psi}$ {\bf setting 2.} $\quad\mathbf{n=50}, \quad\mathbf{p=30}$\vspace{0.3cm}\\
\tiny{
\begin{tabular}{cc|cc|cc|cc}
    \bf{ $\mathbf{h_i}$ setting 1}&&\multicolumn{6}{l}{\bf{ Additive noise: $\sigma_E=0$}}\\
    \hline
  Warping & Nuisance  &\multicolumn{2}{l|}{LISE $\times 10^3$}&\multicolumn{2}{l|}{HMISE$\times 10^3$}&\multicolumn{2}{l}{XMISE}\\
   distortion & distortion & $\hat{\lambda}$& $\hat{\lambda}_{CM}$& $\hat{h}_i$&$\hat{h}_{i,CM}$& $\hat{X}_{ij}$&$\hat{X}_{ij,CM}$\\
   \hline   
    &$\sigma_D\!=\!0$&0.81  ( 0.66 )&\bf 0.52  ( 0.63 )&\bf 0.01  ( 0.02 )&0.05  ( 0.04 )&\bf 0.12  ( 0.07 )&0.35  ( 0.06 )\\
     \cline{2-8}
    $\sigma_W\!=\!0.5$&$\sigma_D\!=\!0.5$&0.65  ( 0.53 )&\bf 0.46  ( 0.55 )&\bf 0.03  ( 0.02 )&0.09  ( 0.05 )&\bf 8.61  ( 0.47 )&8.74  ( 0.46 )\\
     \cline{2-8}
    &$\sigma_D\!=\!1$&\bf 0.85  ( 0.79 )&1.19  ( 0.95 )&\bf 0.28  ( 0.05 )&0.39  ( 0.09 )&\bf 88.47  ( 3.26 )&89.11  ( 3.53 )\\
     \hline
    &$\sigma_D\!=\!0$&\bf 1.17  ( 1.02 )&1.58  ( 1.44 )&\bf 0.39  ( 0.44 )&1.21  ( 0.71 )&\bf 1.67  ( 0.8 )&14.78  ( 7.62 )\\
    \cline{2-8}
    $\phantom{iii}\sigma_W\!=\!1\phantom{iii}$&$\sigma_D\!=\!0.5$&\bf 0.94  ( 0.71 )&1.4  ( 1.13 )&\bf 0.31  ( 0.39 )&1.12  ( 0.69 )&\bf 9.78  ( 0.62 )&22.38  ( 6.79 )\\
    \cline{2-8}
   &$\sigma_D\!=\!1$&\bf 1.13  ( 0.99 )&3.35  ( 1.36 )&\bf 0.49  ( 0.22 )&2  ( 0.59 )&\bf 85.81  ( 3.2 )&106.06  ( 8.68 )\\
    \hline
    \end{tabular}}\vspace{0.3cm}\\
    \tiny{
}\begin{tabular}{cc|cc|cc|cc}
    \bf{ $\mathbf{h_i}$ setting 1}&&\multicolumn{6}{l}{\bf{ Additive noise: $\sigma_E=1$}}\\
    \hline
  Warping & Nuisance  &\multicolumn{2}{l|}{LISE $\times 10^3$}&\multicolumn{2}{l|}{HMISE$\times 10^3$}&\multicolumn{2}{l}{XMISE}\\
   distortion & distortion & $\hat{\lambda}$& $\hat{\lambda}_{CM}$& $\hat{h}_i$&$\hat{h}_{i,CM}$& $\hat{X}_{ij}$&$\hat{X}_{ij,CM}$\\
    \hline
    &$\sigma_D\!=\!0$&
    0.97  ( 0.89 )&\bf 0.5  ( 0.51 )&\bf 0.03  ( 0.02 )&0.04  ( 0.03 )&2.05  ( 0.2 )&\bf 1.54  ( 0.07 )\\
     \cline{2-8}
    $\sigma_W=0.5$&$\sigma_D\!=\!0.5$&0.98  ( 1 )&\bf 0.43  ( 0.54 )&\bf 0.04  ( 0.02 )&0.07  ( 0.04 )&10.44  ( 0.47 )&\bf 9.8  ( 0.45 )\\
     \cline{2-8}
    &$\sigma_D\!=\!1$&\bf 0.9  ( 0.79 )&1.07  ( 0.87 )&\bf 0.29  ( 0.05 )&0.36  ( 0.09 )&92.34  ( 3.61 )&\bf 91.45  ( 3.95 )\\
     \hline
    &$\sigma_D\!=\!0$&\bf 1.09  ( 0.97 )&1.12  ( 0.92 )&\bf 0.3  ( 0.24 )&0.96  ( 0.54 )&\bf 5.01  ( 1.27 )&14.85  ( 6.72 )\\
     \cline{2-8}
    $\phantom{iii}\sigma_W\!=\!1\phantom{iii}$&$\sigma_D\!=\!0.5$&\bf 0.96  ( 1.07 )&1.18  ( 0.95 )&\bf 0.33  ( 0.28 )&1.1  ( 0.69 )&\bf 13.11  ( 1.15 )&23.53  ( 8.21 )\\
    \cline{2-8}
   &$\sigma_D\!=\!1$&\bf 1.37  ( 1.06 )&3.4  ( 1.41 )&\bf 0.62  ( 0.3 )&2.11  ( 0.8 )&\bf 89.67  ( 3.41 )&107.43  ( 9.34 )\\
    \hline
    \end{tabular}\vspace{0.3cm}\\
    \tiny{
}\begin{tabular}{cc|cc|cc|cc}
    \bf{ $\mathbf{h_i}$ setting 2}&&\multicolumn{6}{l}{\bf{ Additive noise: $\sigma_E=0$}}\\
    \hline
  Warping & Nuisance  &\multicolumn{2}{l|}{LISE $\times 10^3$}&\multicolumn{2}{l|}{HMISE$\times 10^3$}&\multicolumn{2}{l}{XMISE}\\
   distortion & distortion & $\hat{\lambda}$& $\hat{\lambda}_{CM}$& $\hat{h}_i$&$\hat{h}_{i,CM}$& $\hat{X}_{ij}$&$\hat{X}_{ij,CM}$\\
    \hline
  &$\sigma_D\!=\!0$&\bf 1.26  ( 1.1 )&1.33  ( 1.1 )&\bf 0.43  ( 0.33 )&0.85  ( 0.45 )&\bf 3.42  ( 1.94 )&10.61  ( 4.69 )\\
     \cline{2-8}
    $\varepsilon_W\!=\!0.005$&$\sigma_D\!=\!0.5$&\bf 1.14  ( 0.93 )&1.21  ( 0.87 )&\bf 0.32  ( 0.19 )&0.85  ( 0.42 )&\bf 12.87  ( 1.56 )&17.69  ( 4.23 )\\
     \cline{2-8}
    &$\sigma_D\!=\!1$&\bf 1.42  ( 1.02 )&2.81  ( 1.12 )&\bf 0.5  ( 0.19 )&1.49  ( 0.45 )&\bf 86.94  ( 3.48 )&96.89  ( 7.56 )\\
     \hline
    &$\sigma_D\!=\!0$&\bf 1.91  ( 1.06 )&3.54  ( 1.37 )&\bf 1.05  ( 0.56 )&3.91  ( 1.45 )&\bf 10.32  ( 4.3 )&51.27  ( 18.67 )\\
     \cline{2-8}
   $\varepsilon_W\!=\!0.0075$&$\sigma_D\!=\!0.5$&\bf 2.54  ( 1.65 )&3.92  ( 1.71 )&\bf 1.05  ( 0.59 )&3.97  ( 1.34 )&\bf 19.36  ( 3.43 )&54.94  ( 15.89 )\\
    \cline{2-8}
   &$\sigma_D\!=\!1$&\bf 2.46  ( 1.29 )&5.94  ( 1.67 )&\bf 1.16  ( 0.5 )&5.04  ( 1.18 )&\bf 87.4  ( 4.19 )&121.12  ( 12.9 )\\
    \hline
    \end{tabular}\vspace{0.3cm}\\
    \tiny{
 \begin{tabular}{cc|cc|cc|cc}
    \bf{ $\mathbf{h_i}$ setting 2}&&\multicolumn{6}{l}{\bf{ Additive noise: $\sigma_E=1$}}\\
    \hline
  Warping & Nuisance  &\multicolumn{2}{l|}{LISE $\times 10^3$}&\multicolumn{2}{l|}{HMISE$\times 10^3$}&\multicolumn{2}{l}{XMISE}\\
   distortion & distortion & $\hat{\lambda}$& $\hat{\lambda}_{CM}$& $\hat{h}_i$&$\hat{h}_{i,CM}$& $\hat{X}_{ij}$&$\hat{X}_{ij,CM}$\\
    \hline
    &$\sigma_D\!=\!0$& 1.25  ( 1.13 )&\bf 1.24  ( 0.99 )&\bf 0.37  ( 0.22 )&0.79  ( 0.41 )&\bf 8.51  ( 2.31 )&11.75  ( 4.18 )\\
     \cline{2-8}
$\varepsilon_W\!=\!0.005$&$\sigma_D\!=\!0.5$&1.2  ( 0.84 )&\bf 1.17  ( 0.69 )&\bf 0.36  ( 0.19 )&0.85  ( 0.39 )&\bf 16.85  ( 1.7 )&19.94  ( 4.57 )\\
     \cline{2-8}
    &$\sigma_D\!=\!1$&\bf 1.58  ( 1 )&2.99  ( 1.27 )&\bf 0.56  ( 0.21 )&1.65  ( 0.52 )&\bf 92.49  ( 3.81 )&100.69  ( 7.63 )\\
     \hline
   &$\sigma_D\!=\!0$&\bf 2.73  ( 2.21 )&3.56  ( 1.58 )&\bf 1.16  ( 0.61 )&3.89  ( 1.4 )&\bf 17.22  ( 5.37 )&51.05  ( 19.23 )\\
     \cline{2-8}
   $\varepsilon_W\!=\!0.0075$&$\sigma_D\!=\!0.5$&\bf 2.26  ( 1.28 )&3.8  ( 1.66 )&\bf 1.11  ( 0.61 )&4.14  ( 1.58 )&\bf 24.9  ( 3.84 )&58.48  ( 20.36 )\\
     \cline{2-8}
    \cline{2-8}&$\sigma_D\!=\!1$&\bf 3.16  ( 2.41 )&5.97  ( 1.66 )&\bf 1.3  ( 0.5 )&5.59  ( 1.47 )\bf&92.48  ( 4.5 )&124.86  ( 16.88 )\\
    \hline
    \end{tabular}}
        \caption{LISE, HMISE and XMISE means and standard deviations of depth-based estaimators ($\hat{\lambda}$, $\hat{h}_i$, $\hat{X}_{ij}$) and those proposed in \cite{latentDef} ($\hat{\lambda}_{CM}$, $\hat{h}_{i,CM}$, $\hat{X}_{ij,CM}$) over $N=100$ simulation runs for $p=30$ random component distortion functions ($\psi$ setting 2) and $n=50$. The lowest integrated error values in each setting are highlighted in bold.}\label{res3}
\end{table}

\begin{table}
$\mathbf{\psi}$ {\bf setting 2.} $\quad\mathbf{n=100}, \quad\mathbf{p=30}$\vspace{0.3cm}\\
\tiny{
\begin{tabular}{cc|cc|cc|cc}
    \bf{ $\mathbf{h_i}$ setting 1}&&\multicolumn{6}{l}{\bf{ Additive noise: $\sigma_E=0$}}\\
    \hline
  Warping & Nuisance  &\multicolumn{2}{l|}{LISE $\times 10^3$}&\multicolumn{2}{l|}{HMISE$\times 10^3$}&\multicolumn{2}{l}{XMISE}\\
   distortion & distortion & $\hat{\lambda}$& $\hat{\lambda}_{CM}$& $\hat{h}_i$&$\hat{h}_{i,CM}$& $\hat{X}_{ij}$&$\hat{X}_{ij,CM}$\\
   \hline   
    &$\sigma_D\!=\!0$&0.78  ( 0.71 )&\bf 0.48  ( 0.56 )&\bf 0.01  ( 0.01 )&0.03  ( 0.02 )&\bf 0.12  ( 0.07 )&0.24  ( 0.02 )\\
     \cline{2-8}
    $\sigma_W\!=\!0.5$&$\sigma_D\!=\!0.5$&0.85  ( 0.79 )&\bf 0.52  ( 0.53 )&\bf 0.02  ( 0.01 )&0.06  ( 0.03 )&\bf 8.43  ( 0.35 )&8.54  ( 0.36 )\\
     \cline{2-8}
    &$\sigma_D\!=\!1$&\bf 0.91  ( 0.94 )&1.04  ( 0.83 )&\bf 0.28  ( 0.03 )&0.32  ( 0.05 )&\bf 88.75  ( 3.19 )&89.92  ( 3.54 )\\
     \hline
    &$\sigma_D\!=\!0$&\bf 0.87  ( 0.7 )&1.17  ( 0.97 )&\bf 0.2  ( 0.18 )&0.91  ( 0.42 )&\bf 1.53  ( 0.47 )&13.1  ( 5.35 )\\
    \cline{2-8}
   $\phantom{iii}\sigma_W\!=\!1\phantom{iii}$&$\sigma_D\!=\!0.5$&\bf 0.95  ( 0.89 )&1.22  ( 1.05 )&\bf 0.2  ( 0.13 )&0.98  ( 0.33 )&\bf 9.46  ( 0.51 )&21.71  ( 4.95 )\\
    \cline{2-8}
   &$\sigma_D\!=\!1$&\bf 0.98  ( 0.86 )&3.02  ( 1.26 )&\bf 0.4  ( 0.11 )&2  ( 0.48 )&\bf 84.99  ( 2.91 )&106.79  ( 7.96 )\\
    \hline
    \end{tabular}}\vspace{0.3cm}\\
    \tiny{
\begin{tabular}{cc|cc|cc|cc}
    \bf{ $\mathbf{h_i}$ setting 1}&&\multicolumn{6}{l}{\bf{ Additive noise: $\sigma_E=1$}}\\
    \hline
  Warping & Nuisance  &\multicolumn{2}{l|}{LISE $\times 10^3$}&\multicolumn{2}{l|}{HMISE$\times 10^3$}&\multicolumn{2}{l}{XMISE}\\
   distortion & distortion & $\hat{\lambda}$& $\hat{\lambda}_{CM}$& $\hat{h}_i$&$\hat{h}_{i,CM}$& $\hat{X}_{ij}$&$\hat{X}_{ij,CM}$\\
    \hline
    &$\sigma_D\!=\!0$&
    0.86  ( 0.83 )&\bf 0.47  ( 0.56 )& 0.02  ( 0.01 )&0.02  ( 0.01 )&1.9  ( 0.16 )&\bf 1.41  ( 0.03 )\\
     \cline{2-8}
    $\sigma_W=0.5$&$\sigma_D\!=\!0.5$&0.82  ( 0.79 )&\bf 0.44  ( 0.53 )&\bf 0.04  ( 0.01 )&0.05  ( 0.02 )&10.29  ( 0.46 )&\bf 9.75  ( 0.41 )\\
     \cline{2-8}
    &$\sigma_D\!=\!1$&\bf 0.88  ( 0.75 )&1.01  ( 0.75 )&\bf 0.28  ( 0.03 )&0.31  ( 0.06 )&\bf 91.25  ( 2.82 )&91.57  ( 3.3 )\\
     \hline
    &$\sigma_D\!=\!0$&\bf 0.98  ( 0.79 )&1  ( 0.79 )&\bf 0.24  ( 0.14 )&0.93  ( 0.41 )&\bf 4.81  ( 0.75 )&15.35  ( 5.36 )\\
     \cline{2-8}
   $\phantom{iii}\sigma_W\!=\!1\phantom{iii}$&$\sigma_D\!=\!0.5$&\bf 1.01  ( 0.94 )&1.05  ( 0.78 )&\bf 0.26  ( 0.19 )&0.97  ( 0.44 )&\bf 12.54  ( 0.73 )&23.69  ( 5.99 )\\
    \cline{2-8}
   &$\sigma_D\!=\!1$&\bf 1.32  ( 0.99 )&2.88  ( 1.15 )&\bf 0.49  ( 0.12 )&1.96  ( 0.42 )&\bf 88.01  ( 3.27 )&106.96  ( 6.71 )\\
    \hline
    \end{tabular}}\vspace{0.3cm}\\
    \tiny{
\begin{tabular}{cc|cc|cc|cc}
    \bf{ $\mathbf{h_i}$ setting 2}&&\multicolumn{6}{l}{\bf{ Additive noise: $\sigma_E=0$}}\\
    \hline
  Warping & Nuisance  &\multicolumn{2}{l|}{LISE $\times 10^3$}&\multicolumn{2}{l|}{HMISE$\times 10^3$}&\multicolumn{2}{l}{XMISE}\\
   distortion & distortion & $\hat{\lambda}$& $\hat{\lambda}_{CM}$& $\hat{h}_i$&$\hat{h}_{i,CM}$& $\hat{X}_{ij}$&$\hat{X}_{ij,CM}$\\
    \hline
  &$\sigma_D\!=\!0$&\bf 1.03  ( 0.63 )&1.13  ( 0.73 )&\bf 0.28  ( 0.14 )&0.66  ( 0.27 )&\bf 3  ( 1.39 )&9.06  ( 3.1 )\\
     \cline{2-8}
    $\varepsilon_W\!=\!0.005$&$\sigma_D\!=\!0.5$&\bf 1.03  ( 0.81 )&1.2  ( 1.06 )&\bf 0.23  ( 0.12 )&0.69  ( 0.22 )&\bf 11.8  ( 0.89 )&16.93  ( 3 )\\
     \cline{2-8}
    &$\sigma_D\!=\!1$&\bf1.41  ( 0.99 )&2.65  ( 0.92 )&\bf0.44  ( 0.09 )&1.51  ( 0.31 )&\bf86.42  ( 2.89 )&98.37  ( 5.99 )\\
     \hline
    &$\sigma_D\!=\!0$&\bf2.08  ( 1.12 )&3.33  ( 1.12 )&\bf0.85  ( 0.29 )&3.66  ( 0.92 )&\bf9.99  ( 3.4 )&50.22  ( 13.01 )\\
     \cline{2-8}
   $\varepsilon_W\!=\!0.0075$&$\sigma_D\!=\!0.5$&\bf2.17  ( 1.1 )&3.44  ( 1.21 )&\bf0.83  ( 0.29 )&3.84  ( 0.99 )&\bf18.41  ( 2.55 )&56.11  ( 13.19 )\\
    \cline{2-8}
   &$\sigma_D\!=\!1$&\bf2.41  ( 1.22 )&5.73  ( 1.31 )&\bf0.96  ( 0.21 )&5.41  ( 1 )&\bf85.47  ( 2.79 )&125.64  ( 12.13 )\\
    \hline
    \end{tabular}}\vspace{0.3cm}\\
    \tiny{
 \begin{tabular}{cc|cc|cc|cc}
    \bf{ $\mathbf{h_i}$ setting 2}&&\multicolumn{6}{l}{\bf{ Additive noise: $\sigma_E=1$}}\\
    \hline
  Warping & Nuisance  &\multicolumn{2}{l|}{LISE $\times 10^3$}&\multicolumn{2}{l|}{HMISE$\times 10^3$}&\multicolumn{2}{l}{XMISE}\\
   distortion & distortion & $\hat{\lambda}$& $\hat{\lambda}_{CM}$& $\hat{h}_i$&$\hat{h}_{i,CM}$& $\hat{X}_{ij}$&$\hat{X}_{ij,CM}$\\
    \hline
 &$\sigma_D\!=\!0$& 1.17  ( 0.71 )&\bf 1.01  ( 0.62 )&\bf 0.29  ( 0.11 )&0.71  ( 0.25 )&\bf 7.79  ( 1.67 )&11  ( 3.02 )\\
     \cline{2-8}
$\varepsilon_W\!=\!0.005$&$\sigma_D\!=\!0.5$&1.08  ( 0.71 )&\bf1.05  ( 0.56 )&\bf0.29  ( 0.13 )&0.72  ( 0.24 )&\bf15.79  ( 1.23 )&19.19  ( 3.09 )\\
     \cline{2-8}
    &$\sigma_D\!=\!1$&\bf1.57  ( 0.95 )&2.68  ( 0.95 )&\bf0.5  ( 0.1 )&1.51  ( 0.31 )&\bf90.47  ( 3.15 )&100.22  ( 5.95 )\\
     \hline
   &$\sigma_D\!=\!0$&\bf1.91  ( 0.93 )&3.29  ( 0.98 )&\bf0.91  ( 0.27 )&3.8  ( 0.87 )&\bf16.73  ( 3.68 )&53.55  ( 12.68 )\\
     \cline{2-8}
   $\varepsilon_W\!=\!0.0075$&$\sigma_D\!=\!0.5$&\bf2.06  ( 1.47 )&3.51  ( 1.35 )&\bf0.9  ( 0.24 )&4.08  ( 1.03 )&\bf23.68  ( 2.44 )&62.24  ( 14.76 )\\
     \cline{2-8}
    \cline{2-8}&$\sigma_D\!=\!1$&\bf2.48  ( 1.22 )&5.62  ( 1.44 )&\bf1.1  ( 0.23 )&5.16  ( 1.06 )&\bf90.76  ( 2.88 )&124.83  ( 11.52 )\\
    \hline
    \end{tabular}}
     \caption{LISE, HMISE and XMISE means and standard deviations of depth-based estaimators ($\hat{\lambda}$, $\hat{h}_i$, $\hat{X}_{ij}$) and those proposed in \cite{latentDef} ($\hat{\lambda}_{CM}$, $\hat{h}_{i,CM}$, $\hat{X}_{ij,CM}$) over $N=100$ simulation runs for $p=30$ random component distortion functions ($\psi$ setting 2) and $n=100$. The lowest integrated error values in each setting are highlighted in bold.}\label{res4}
\end{table}

\section{Analysis of real data}\label{appli}

\subsection{Phase variability in sea ice extent in the Arctic Ocean}

There has been long time evidence about the relation between the loss of ice masses in the Arctic and Antarctic oceans and global warming \citep{IceGlobalWarming1, IceGlobalWarming2, IceGlobalWarming3}. In this section we analyse annual curves of daily sea ice extent values on eight different regions of the Arctic ocean, from 1979 to 2024 (Fig. \ref{sie}). We use a public data set from the Norwegian Meteorological Institute \citep{DataSetSIE} in which sea ice extent is provided globally, for the two hemispheres and for specific regions. For this analysis we consider the arctic regions presented in Figure \ref{ocean}: Beaufort Sea, Chukchi Sea, East Siberian Sea, Laptev Sea, Kara Sea, Barents Sea, Fram Strait/Greenland Sea and Svalbard region. The sea ice extent of an ocean region is defined as the area of that region which is covered by a significant amount of sea ice. More precisely, for some given grid, sea ice extent is obtained as the total area of all grid cells with more than 15\% of sea ice concentration \citep{OSI_SAF_manual}. For the data set at hand, this is calculated in a 10km grid. 

\begin{figure}
\begin{center}
\hspace*{-1.5cm}\includegraphics[width=18cm]{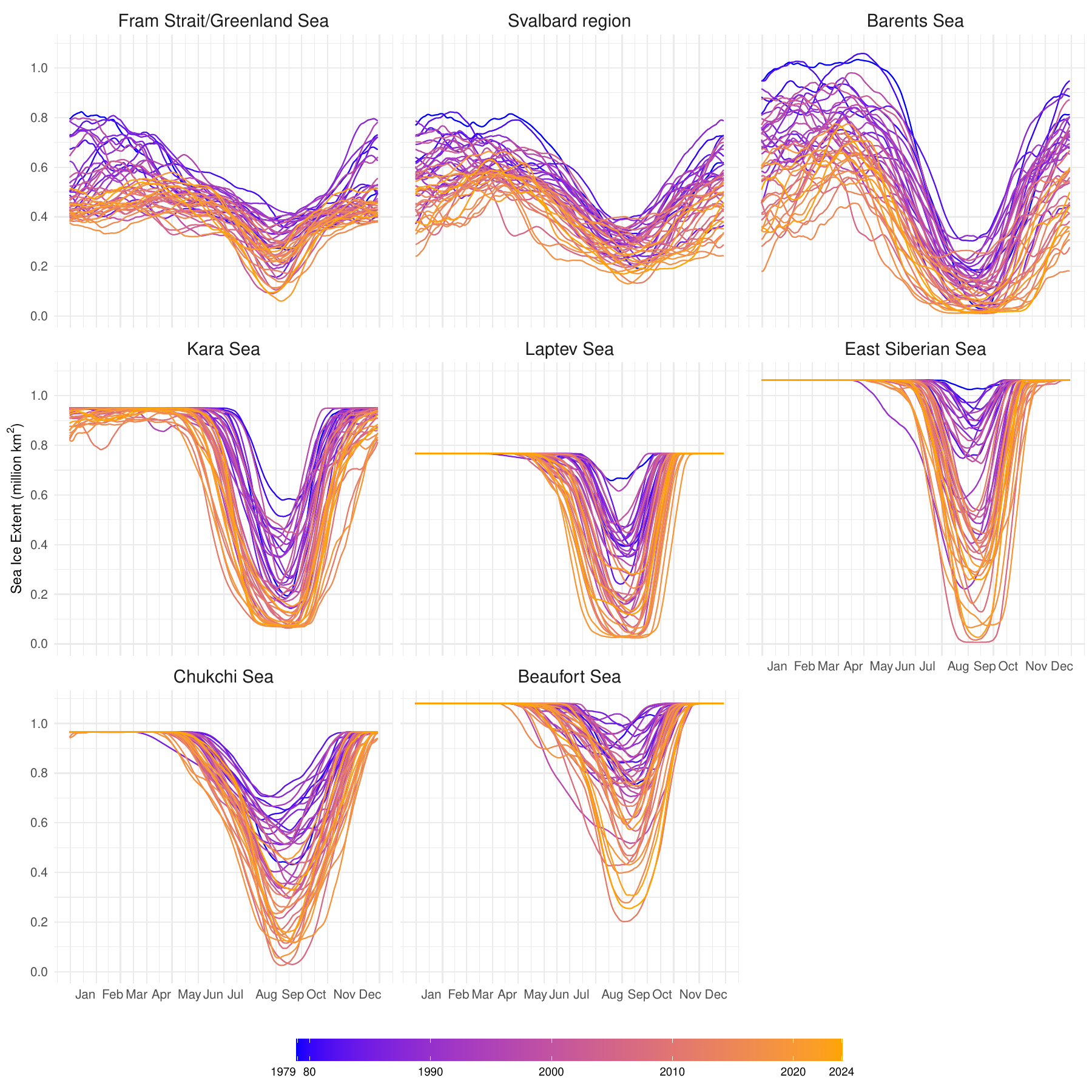}
\end{center}
\caption{Daily sea ice extent values in the eight arctic regions defined in Figure \ref{ocean}, from 1979 to 2024.}\label{sie}
\end{figure}

\begin{figure}
\begin{center}
\hspace*{-1.5cm}\includegraphics[width=10cm]{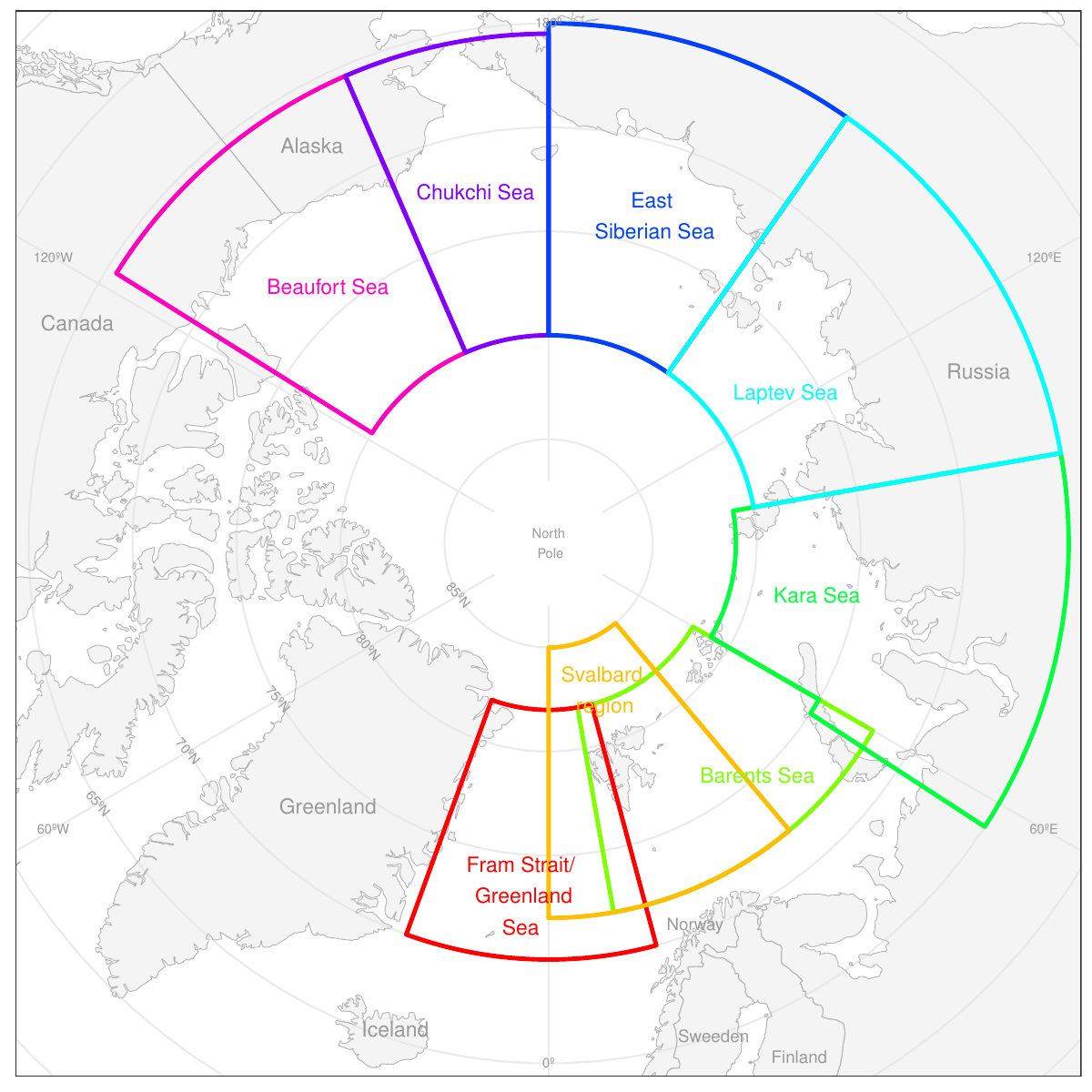}
\end{center}
\caption{Bounding boxes for the the calculation of regional sea-ice indexes in the Arctic Ocean, as defined by \cite{OSI_SAF_manual}.}\label{ocean}
\end{figure}

Daily sea ice extent values were available from October 25th 1978 to the present date (2025), but we only kept whole years, that is, data from 1979 to 2024. We had to handle missing data. In particular, the same set of 156 daily records were missing from all regions on years 1984, 1986 to 1988 and 2024. They were imputed by means of linear interpolation. Also, a moving average smoothing with a window of 15 days was performed to reduce noise, for which data from the end of 1978 and the beginning of 2025 were used. The resulting annual curves are presented in Figure \ref{sie}. It is clear that ice extent has globally decreased over time in the Arctic ocean, with changes visible both in amplitude, the area covered by ice, which is lower each year, and phase, the periods of the year in which this area is at is maximum value, which are shorter each year. Therefore, there is important phase variation over time within each region, and also between the annual patterns of each region. However, we can observe two different ice extent behaviours characterizing two groups of seas: those from approximately 60ºE to 120ºW longitude (Kara, Laptev, East Siberian, Chukchi and Beaufort) exhibiting a constant pattern over winter and spring followed by a quick decay in summer and the subsequent ascent during autumn; and those in longitudes 30ºW to 60ºE (Greenland, Barents, Svalbard), which show a more variable annual pattern with a maximum around spring and a minimum around late summer. Because of this difference in behaviours, we cannot assume that there is a common underlying ice extent amplitude function for all the regions, and therefore we will analyse the two groups separately. 

We apply our depth-based estimation procedure for the latent deformation model to both sets of five and three ocean regions, where each individual region is a component of the process and each annual curve represents an individual observation. Warping estimates obtained for each set of multivariate curves are shown in Figure \ref{sie_warps}. In the group of five seas, due to the constant nature of ice extent during winter months warping functions are estimated with a steep, nearly vertical, segment at the begining and end of the year. We can also appreciate a distinct behaviour of the warping functions corresponding to more recent years, in comparison to the 80's and 90's of the previous century.  From around 2010, the annual ice extent process is accelerated in spring and delayed in autumn, corresponding to longer periods of lower ice concentration. In the group of the Greenland sea, Barents sea and Svalbard region, warping estimates are more irregular, as the observed curves themselves, and the evolution of phase over time is not that clear.\\
In order to detect possible anomalies in phase, we apply functional outlier detection techniques to the warping estimates. We also look for marginal and joint outliers in the original multivariate curves, to identify those years which are genuinely \emph{phase} outliers and do not present other types of anomalies. In particular, we use the functional boxplot \citep{fbplot} and the outliergram \citep{outgram}, to look for both amplitude and shape outliers, on the warping estimates and marginally in each variable of the observed multivariate sample. For joint outlier detection, we use the Depthgram \citep{Depthgram}, functional outlier map (FOM, \cite{FOM}) and magnitude-shape plot(MS-plot, \cite{MSplot}, on the multivariate ice-extent curves. In the five-seas region from Beaufort to Kara sea, no warping outlier has been detected. An outlier search on the original curves themselves does not reveal any interesting finding. In fact, we find a large number of marginal atypical curves, in particular many magnitude outliers,  which is common in functional samples for which amplitude variation is close to zero in some subinterval of the time domain. However, none of these potential outliers presents a particularly atypical phase as per the estimated warping functions, as mentioned above. When analysing the group of Greenland,
Barents, and Svalbard sees, the outlier detection procedures on the warping estimates yielded three atypical years: 2016, 2018 and 2022 (Figure \ref{3seas_outs}).  The marginal and joint outlier detection search on the original multivariate curves found atypical annual ice-extent profiles in 1979, 1999, 2003, 2004, 2012, 2014 and 2020, and therefore none of the years identified from the set of warping functions. For these three years, 2016, 2018 and 2022, we can see in Figure \ref{3seas_outs} how the corresponding warping estimates show a \emph{summer bump} indicating a longer than usual low-ice-extent regime, from approximately June to the end of the year.

In this way, we show that a detailed analysis of the warping functions estimated in the latent deformation model can help obtain insights about phase variation that might not be directly observable from the original curves themselves.

\begin{figure}
\begin{center}
\hspace*{-0.75cm}
\includegraphics[width=8cm]{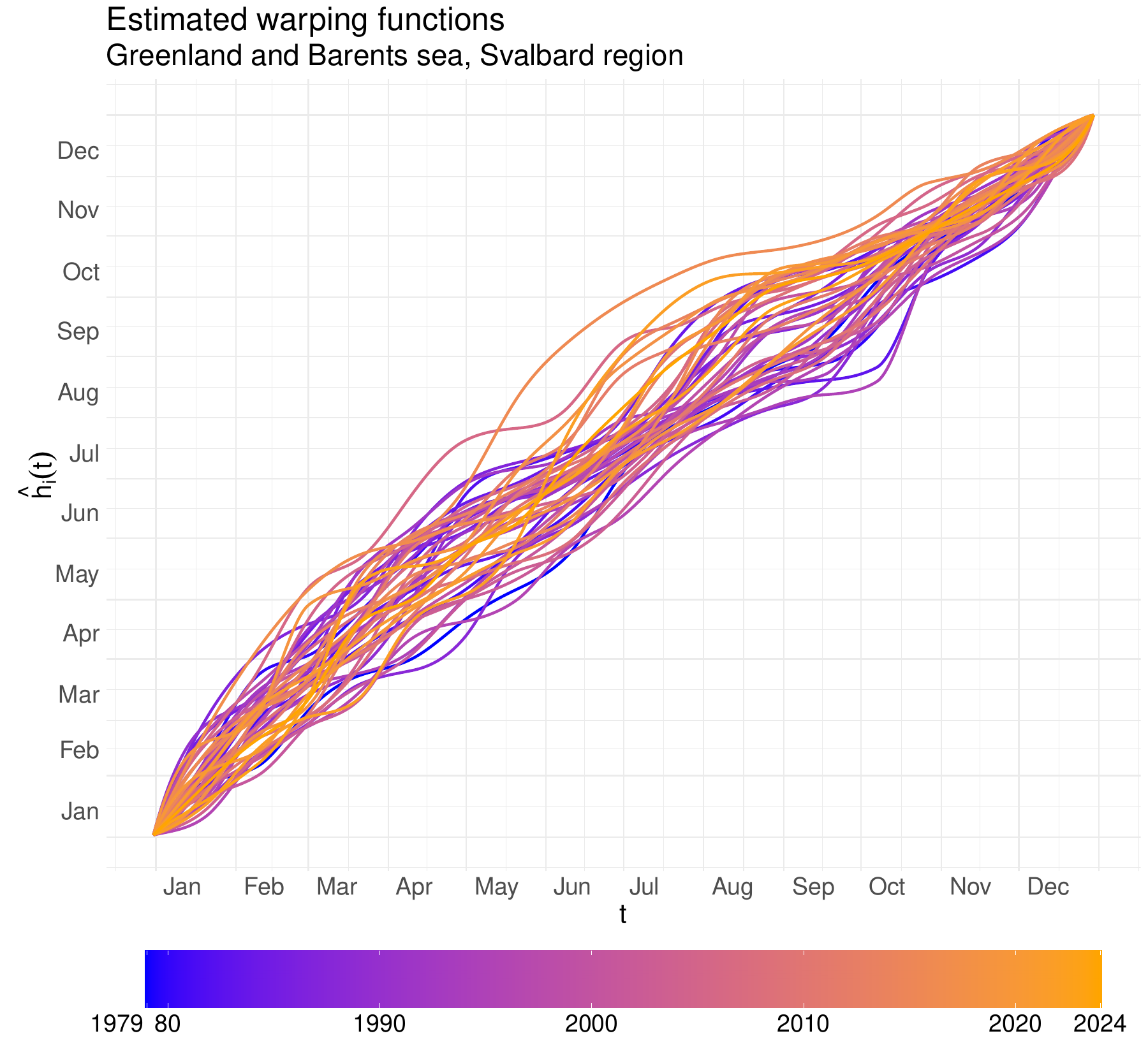}\hspace*{0.1cm}
\includegraphics[width=8cm]{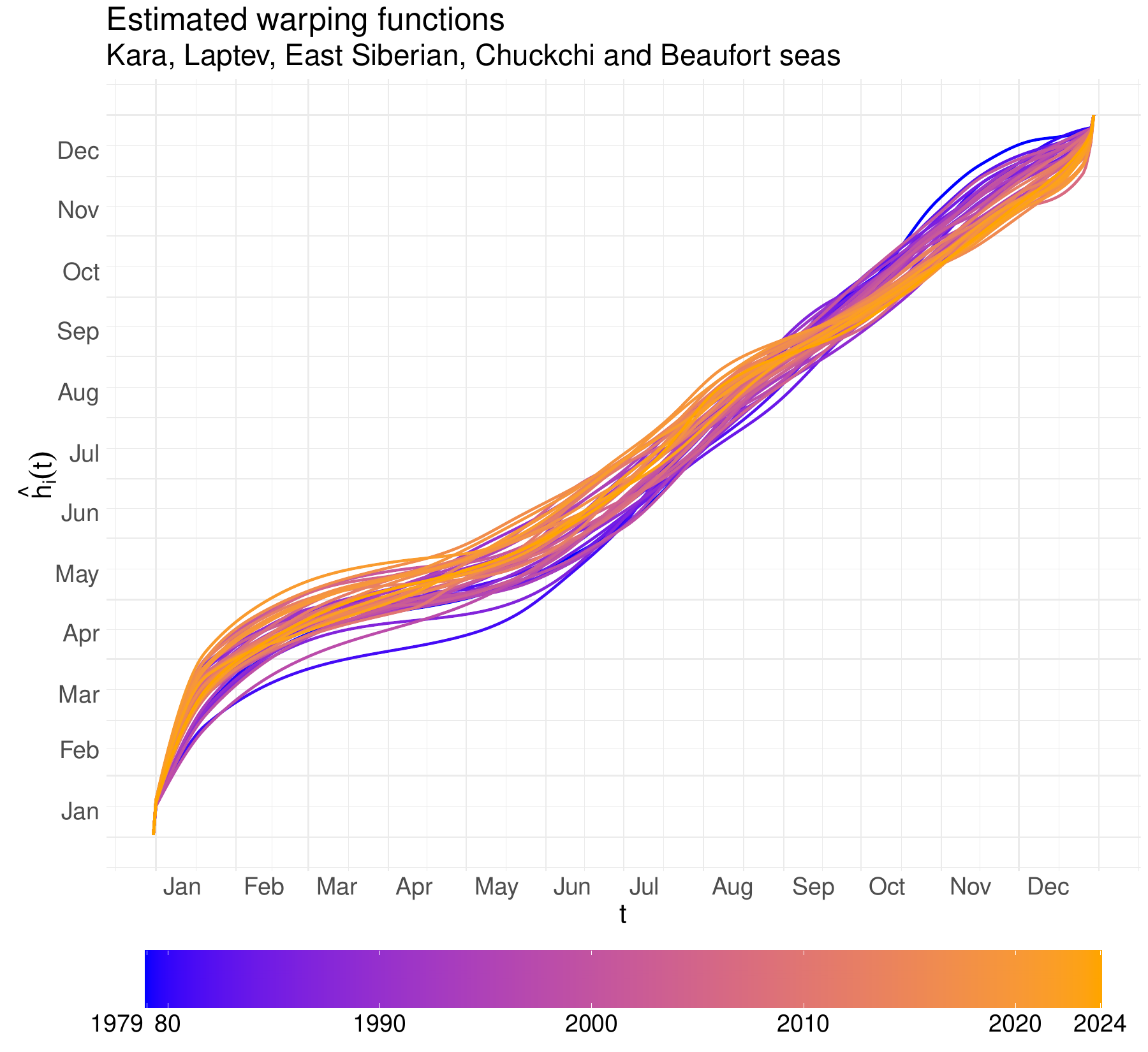}
\end{center}
\caption{Yearly warping functions estimated separately on the two sets of Artic regions.}\label{sie_warps}
\end{figure}

\begin{figure}
\begin{center}
\hspace*{-0.75cm}
\includegraphics[width=16cm,trim = 3.6cm 0.3cm 3.6cm 0.3cm, clip]{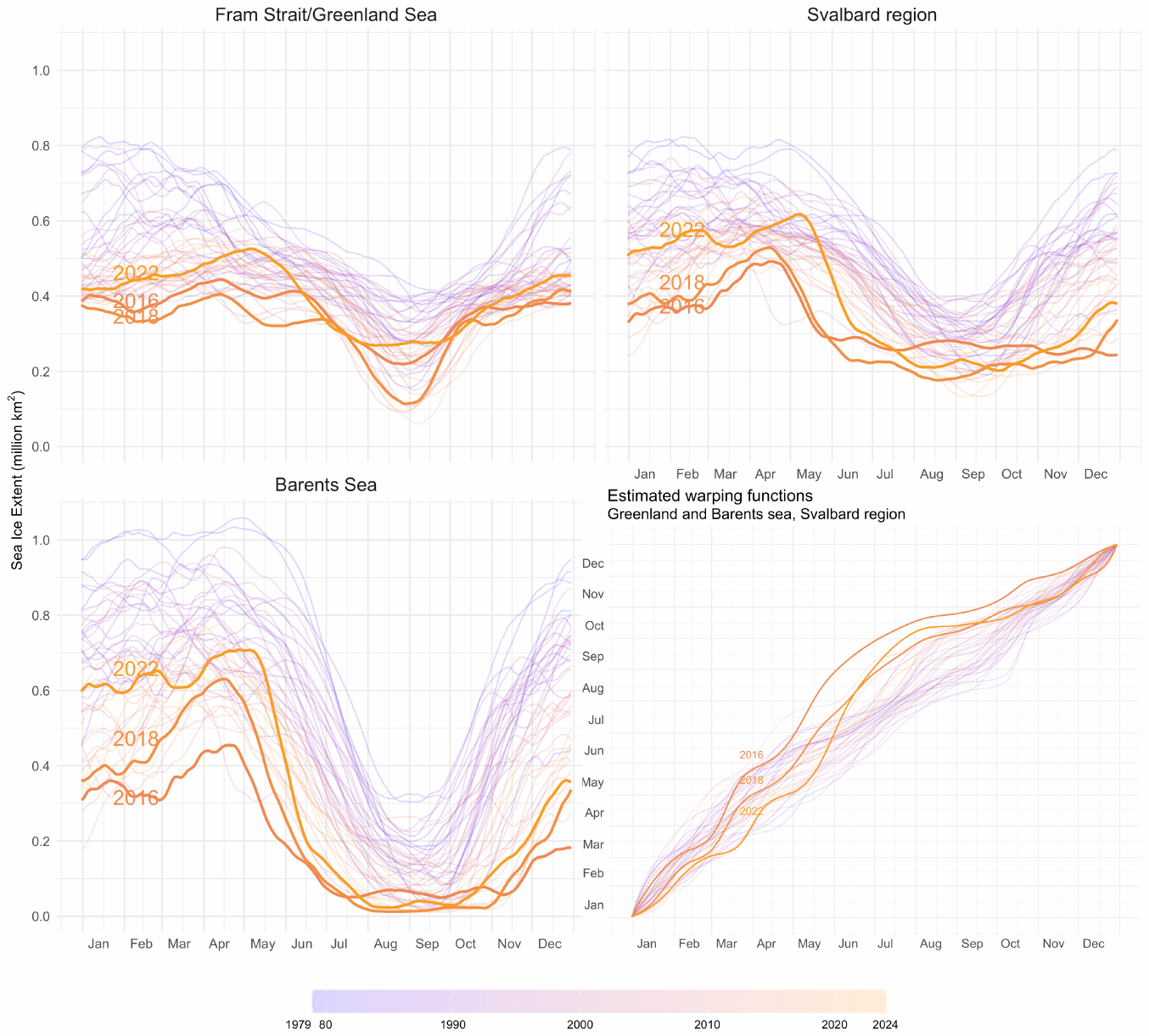}
\end{center}
\caption{Daily ice extent (1980 - 2024) in Greenland Sea, Barents Sea and Svalbard region and corresponding warping functions. Years 2016, 2018 and 2022 are detected as having an atypical warping profile.}\label{3seas_outs}
\end{figure}

\subsection{Mother age distribution across Europe}\label{data2}

In this section we analyse the evolution of maternity across Europe from the beginning of the 21st century. In particular, we focus on the distribution of mothers age at the birth of their first child for a set of 27 European countries, in the years 2000, 2010 and 2020. The countries considered are those among a larger group of 47 countries included in the Eurostat dataset \citep{DataSetMaternity} for which the target variable was available at all three years. These are: Austria, Belgium, Bulgaria, Switzerland, Cyprus, Czechia, Estonia, Greece, Spain, Finland, Croatia, Hungary, Ireland, Iceland, Lithuania, Luxembourg, Latvia, North Macedonia, Netherlands, Norway, Poland, Portugal, Romania, Serbia, Sweden, Slovenia, Slovakia. For each country and year, counts of women giving birth to their first child were reported, grouped by women age in 1 year intervals, from 15 to 49 years old. Those were used to obtain density curves by smoothing the raw data by means of cubic B-splines. The resulting multivariate data set is presented in Figure \ref{Mat_data} (top). It can be seen that all three years share a unimodal age distribution (with one atypical bimodal curve in 2000 corresponding to Ireland). However, in 2000 the distribution is right skewed, and in the following years it becomes roughly symmetrical. The mode varies from around 23 years old in 2000 to around 30 in 2020. Therefore, we fall under the framework of the latent deformation model, in which the different components of the process are the age distributions in different years which can be understood as warped versions of some latent maternity age distribution. Also, within each year, each country distribution presents individual time distortion.\\

It is interesting to understand how maternity age distribution has evolved over time, globally in Europe and in different countries, other than by simply quantifying how much the mean, median or modal maternity age has increased during the last 20 years.
In Figure \ref{Mat_data} (center) we see all three estimated amplitude functions using our proposed depth-based method,  corresponding to Cyprus in 2000, Croatia in 2010 and Estonia in 2020, as well as the global amplitude function, which in this case is again the age distribution in Croatia 2010. We can appreciate how the mode and the shape of the first-time-mothers age distribution changes over time. From the warping functions estimated globally over all three components (Fig \ref{Mat_data} center right), we see that country time distortions seem to have a quite regular and symmetrical structure around the diagonal of the age interval. However, if we look at the warping functions estimated on each component (Fig \ref{Mat_data} bottom), it is clear that the time distortion variability of the maternity age distribution has decreased from 2000 to 2020, when there is a higher homogeneity between the set of analysed countries. Warping estimates can also help highlight some features in the data: for instance, the 2010 warping with noticeable lowest values between 20 and 30 years old corresponds to Spain and its significantly late maternity age (see Fig \ref{Mat_count} bottom line).
Another relevant question is whether each country's evolution towards a later entry into maternity is faster or slower than that of neighbouring countries. In other words, a country with younger/older mothers in 2000 is also a country with younger/older mothers in 2020, in comparison with the rest? A way of analysing this is by looking at the correspondence between the warping estimates obtained for each one of the three years. If early (resp. late) maternity is preserved across time for some country, its corresponding warping estimates should have relatively high/over the diagonal (resp. low/under the diagonal) values for all 2000, 2010 and 2020. This can be visualized in the last row of Figure \ref{Mat_data}, where the color scale used for each country corresponds to the hypograph values of the warping estimates in the 2000 sample. By further representing the correspondence between the hypograph values across the three samples of warping estimates with the WHyRA plot (Fig \ref{Mat_count} top) we can see which countries have consistently low, medium or high hypograph values, and which ones show visible changes over the three years. In general, we can say that in all three scatter plots points are cluttered close to the diagonal line, indicating a high consistency over time. However, we can observe changes for some countries, specially noticeable in the scatter plot comparing the warpings' hypograph values of 2000 and 2020. For instance, the point corresponding to Romania lies below the diagonal and has a high hypograph value in 2000 and a medium value for 2020.  This indicates an early maternity age in 2000 which has become an "average" maternity age distribution in 2020 (Fig \ref{Mat_count} second row). The opposite occurs with Sweden (Fig \ref{Mat_count} third row) which used to have one of the latest maternity age distribution which is now quite standard. As an example of a country whose relative maternity timing with respect to the rest of countries hasn't changed over time there is Spain (Fig \ref{Mat_count} bottom), whose corresponding point in the hypograph's scatter plots lies close to the diagonal.

\begin{figure}
\hspace*{-1cm}\includegraphics[width=15cm,trim={0 0.1cm 0 0.7cm},clip]{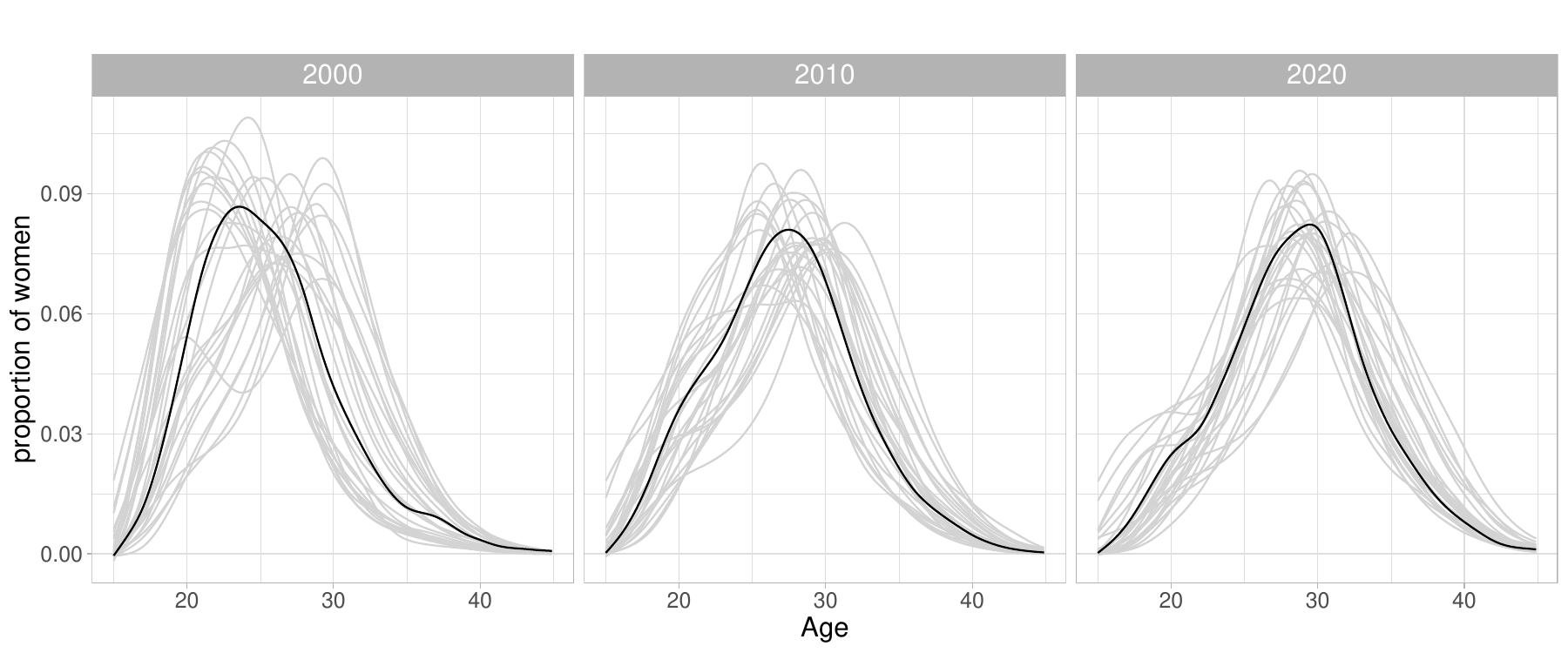}\\
\hspace*{-1cm}\includegraphics[width=7.75cm,trim={0 0.1cm 0 0.7cm},clip]{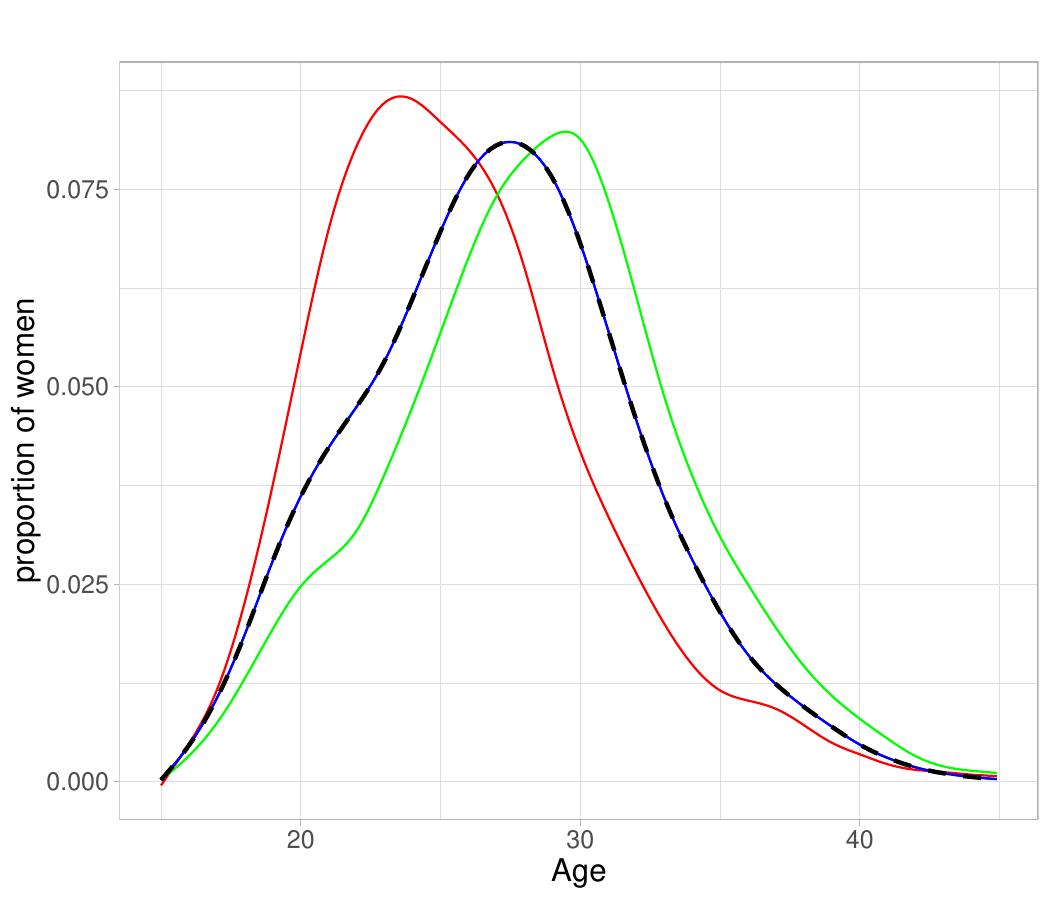}\hspace*{0.8cm}\includegraphics[width=6.6cm,trim={0 0.1cm 0 0.7cm},clip]{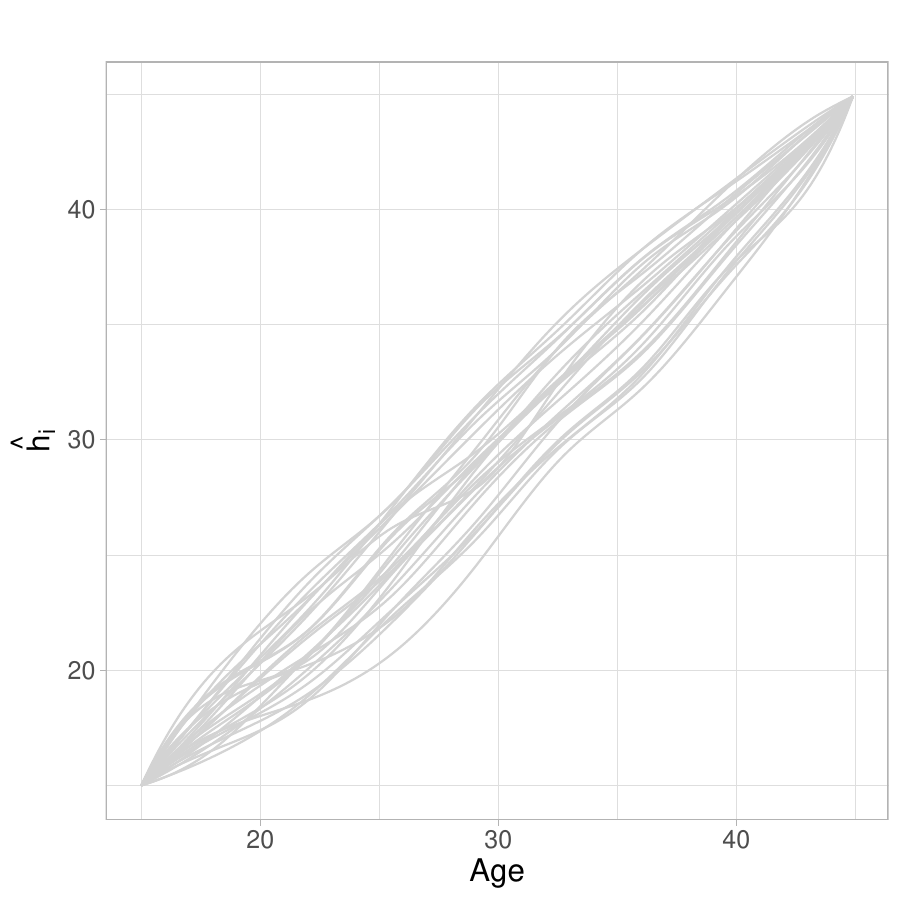}\\
\hspace*{-1cm}\includegraphics[width=14.2cm,trim={0 0.1cm 0 0.7cm},clip]{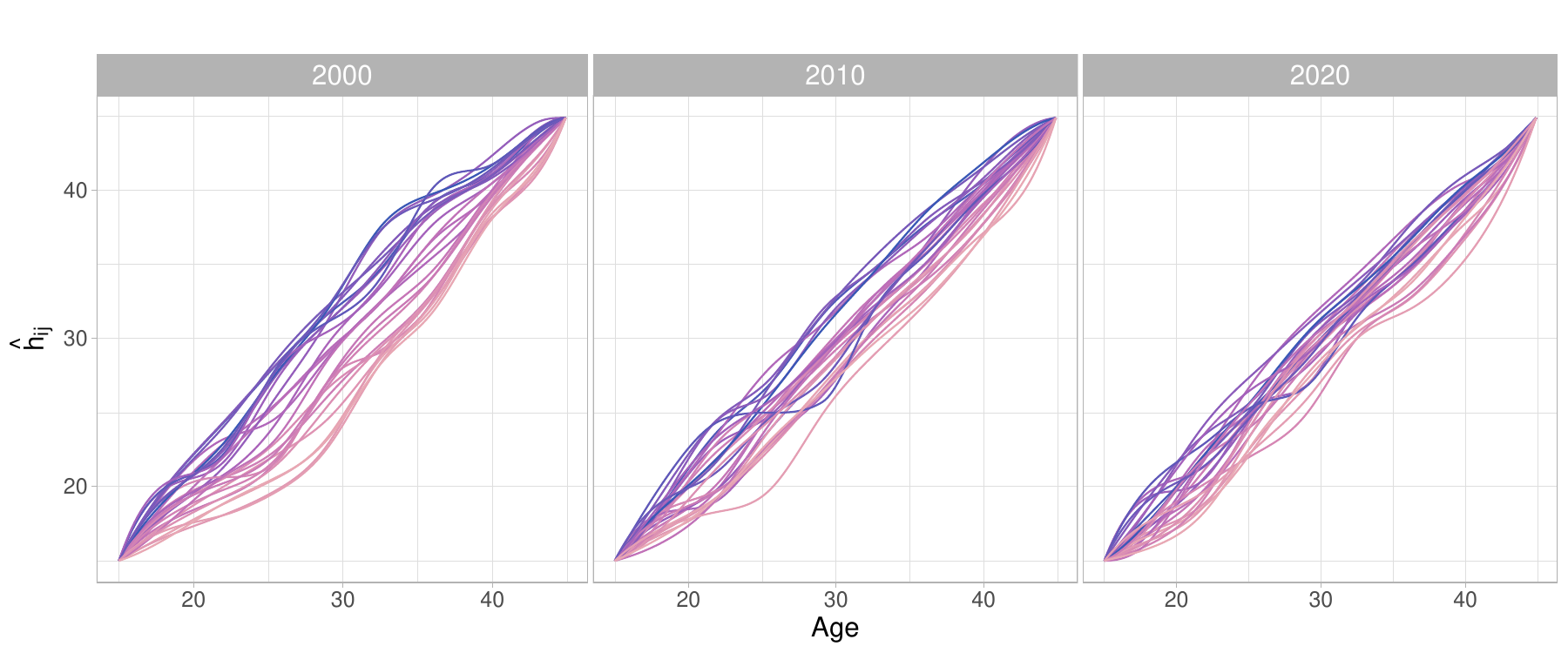}\includegraphics[width=2.6cm,trim={7cm -1cm 0cm 2cm},clip]{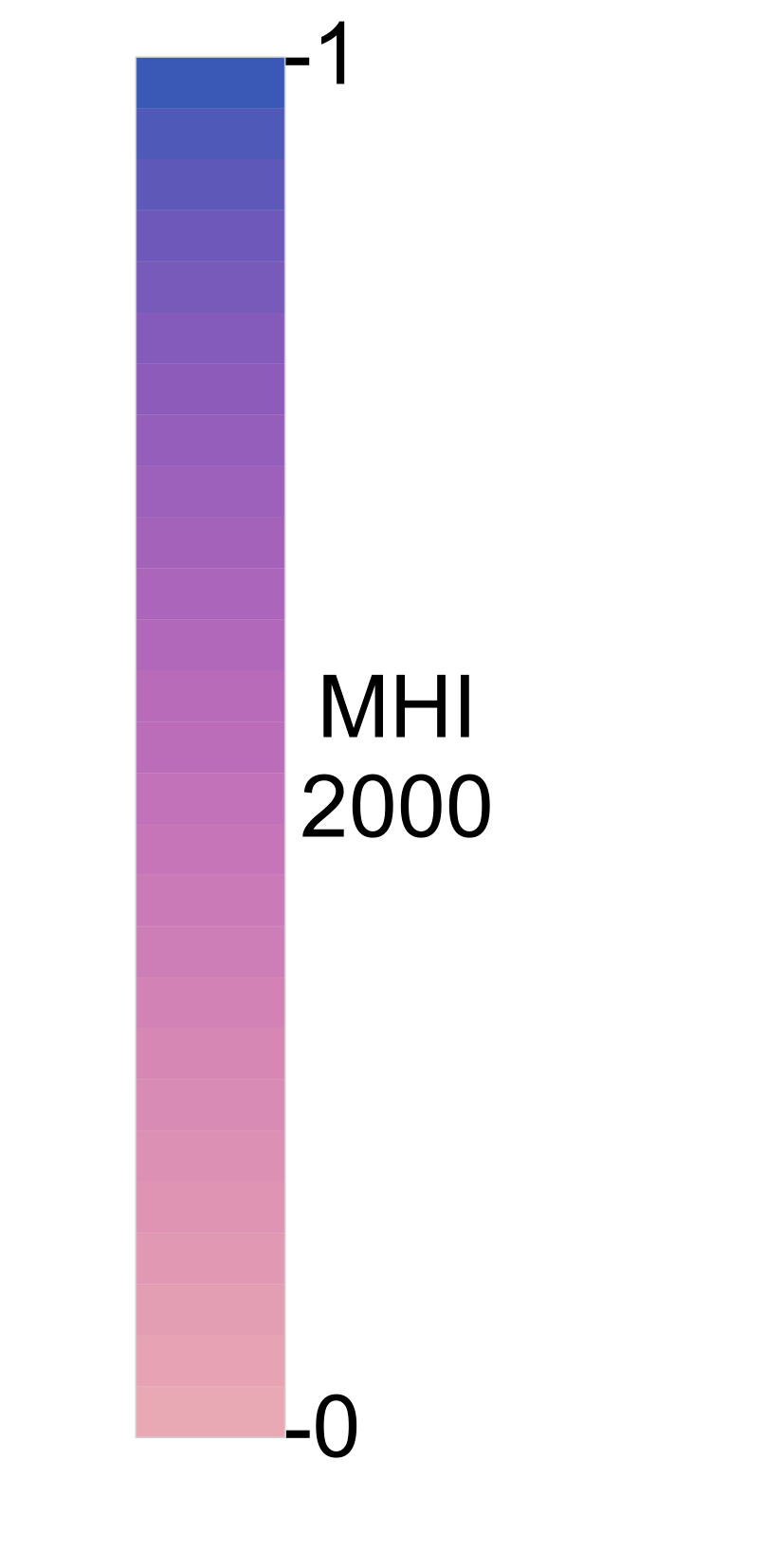}
\caption{Top: Mother age distribution at first child birth for the group of 27 European countries in 2000, 2010 and 2020. The median curve age distribution for each is represented as a black solid curve, and corresponds to Cyprus, Croatia and Estonia, respectively. Center left: Estimated components' amplitude functions ($\hat{\gamma_j}$, coloured solid lines),  and estimated common amplitude function ($\hat{\lambda}$, dashed black line).  Center right: Estimated individual warping functions. Bottom: Individual warping functions estimated on each component. Color code is defined as the modified hypograph value for warping function estimate in the year 2000 sample.}\label{Mat_data}
\end{figure}

\begin{figure} 
\begin{center}
\hspace*{-0.8cm}\includegraphics[width=15cm]{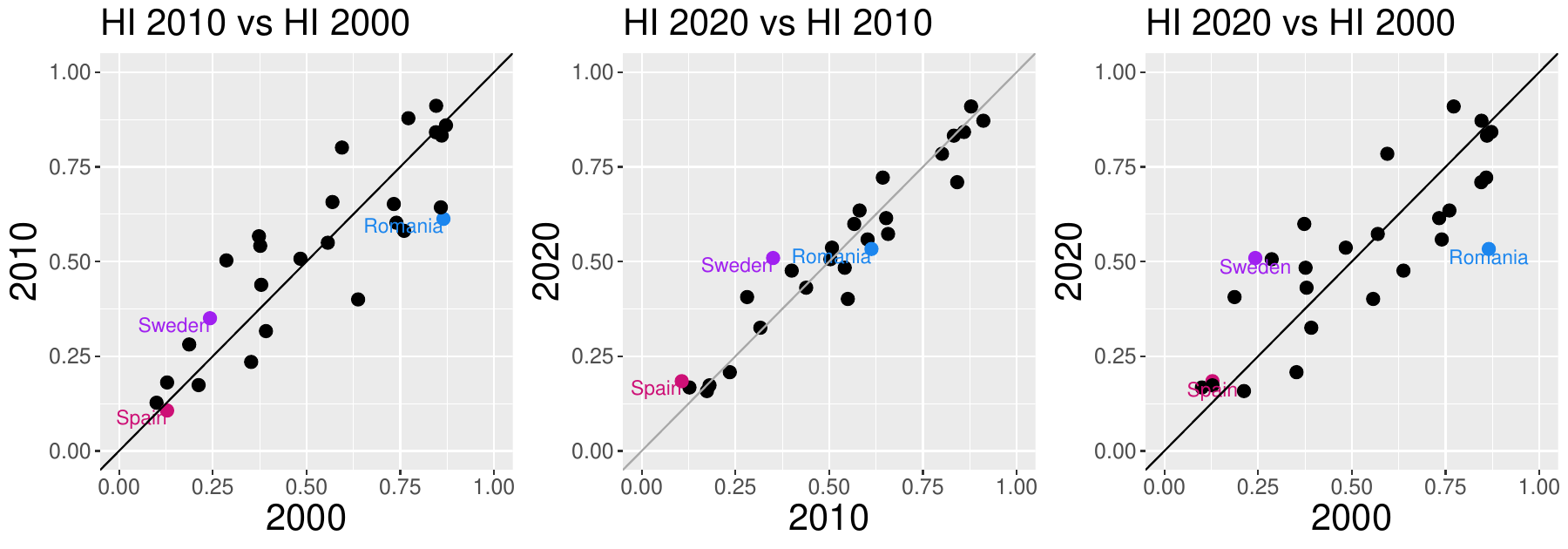}\\
\hspace*{-0.6cm}\includegraphics[width=14.6cm, trim=0 0 0 1cm,clip]{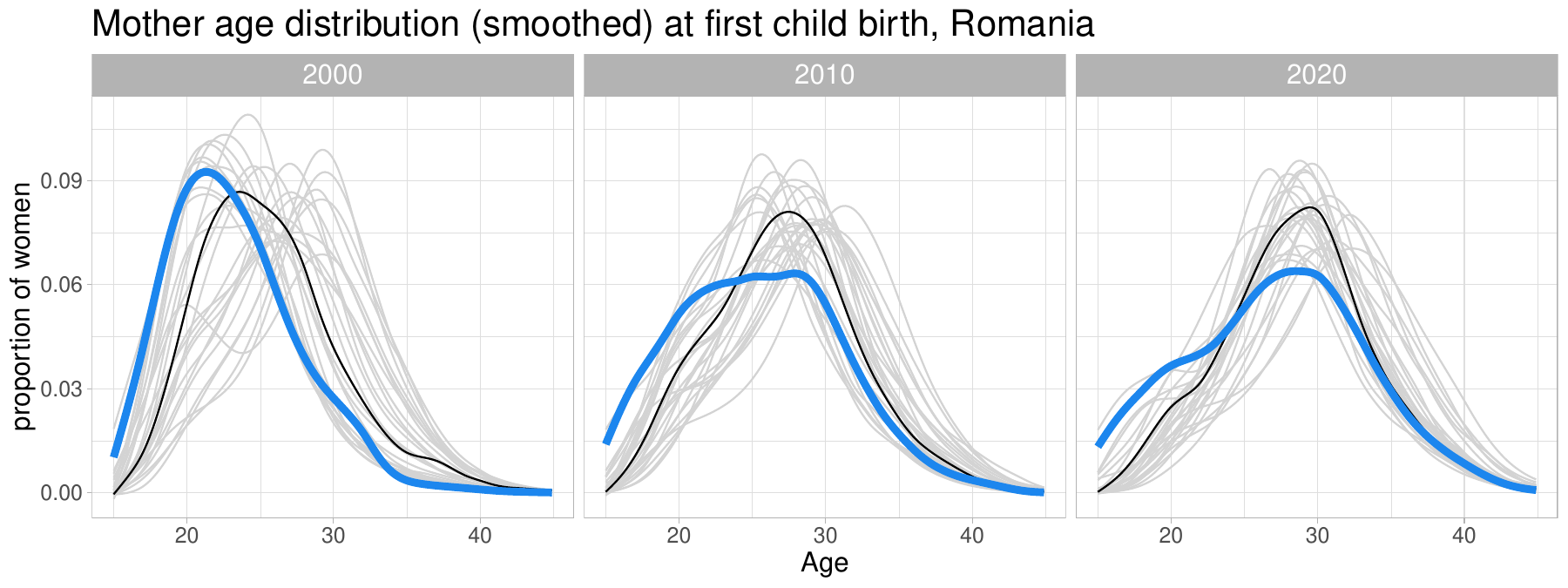}$\,\,$\rotatebox[x=3em,y=6em]{-90}{\texttt{Romania}}\\
\hspace*{-0.6cm}\includegraphics[width=14.6cm, trim=0 0 0 1cm,clip]{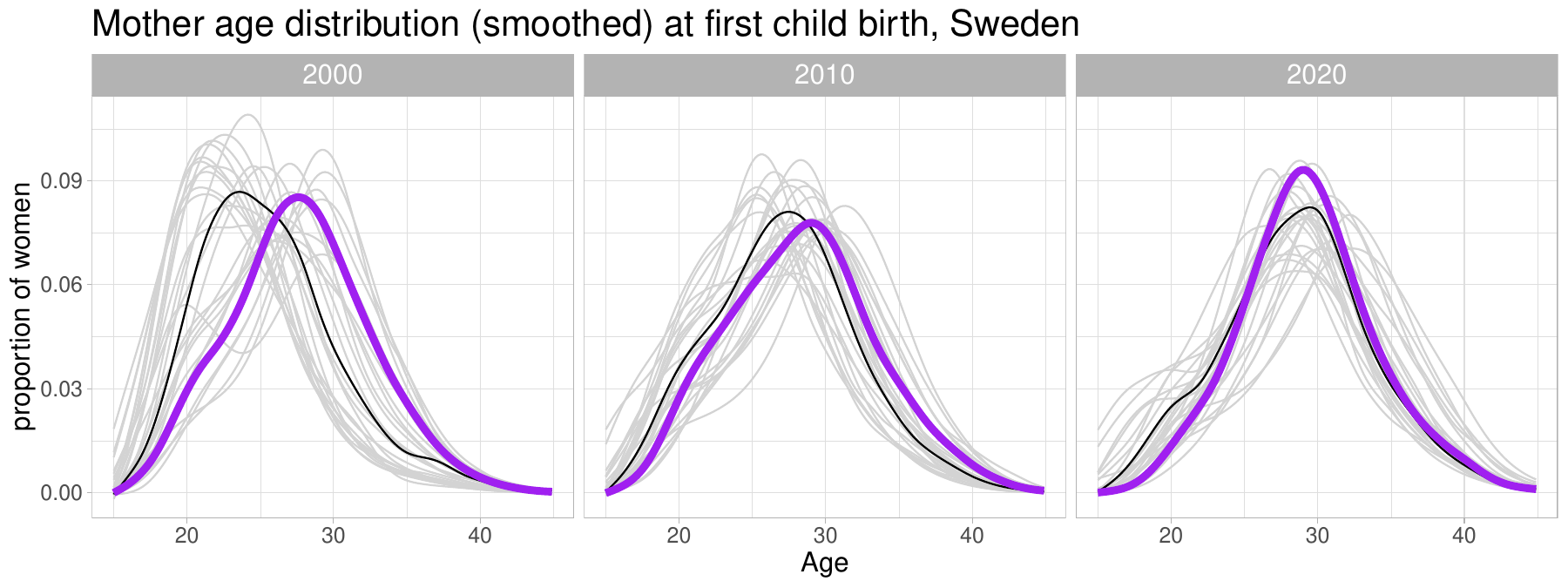}$\,\,$\rotatebox[x=3em,y=6em]{-90}{\texttt{Sweden}}\\
\hspace*{-0.6cm}\includegraphics[width=14.6cm, trim=0 0 0 1cm,clip]{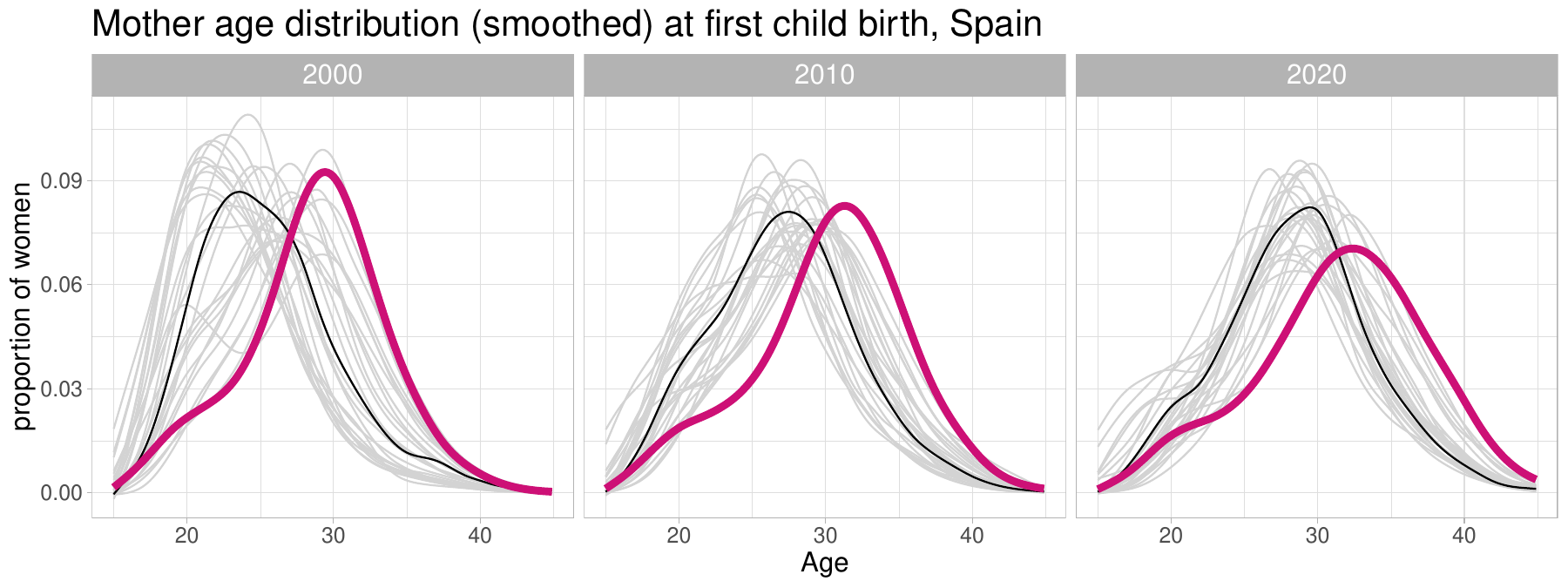}$\,\,$\rotatebox[x=3em,y=6em]{-90}{\texttt{Spain}}
\end{center}
\caption{Top: WHyRA plot of the estimated warping functions in 2000, 2010 and 2020. Bottom: Mother age distribution at first child birth for Romania (blue), Sweden (purple) and Spain (red). The median curve age distribution is represented as a black solid curve for each year.}\label{Mat_count}
\end{figure}

\section{Discussion}\label{disc}
In this article, we propose an alternative estimation method for the latent deformation model \citep{latentDef}, a specific multivariate functional time-warping model in which individual phase variability and cross-component phase variability are separable. Our estimation strategy differs from the original approach in that it is based on functional depth measures and does not require curve registration to estimate the common amplitude function, making it more computationally efficient.

Specifically, we establish conditions under which a general functional depth measure yields a consistent estimator of the target function by directly selecting the deepest curve in the sample (or in the monotonized sample when the target function is non-monotonic). We also discuss the conditions required for the cross-component and individual warping processes so that their composition satisfies the location assumption necessary for consistency of our depth-based estimator.

Although these conditions are strong, the simulation study—where the generating warping processes did not satisfy them—shows that our approach is robust to their violation, as well as to departures from other model assumptions. The results further demonstrate that our method provides robust estimates of all model components in the presence of atypical observations. The proposed methodology is not intended for highly noisy or irregular data, for which a pre-smoothing step would be required.

Finally, we introduce the WHyRA plot, which facilitates model diagnostics and enhances understanding of individual variability in practical applications, as illustrated using real data.

\section*{Appendix}
\subsection*{A1. Proofs of Propositions 1, 2 and 3.}
\emph{Proof of Proposition \ref{prop_univ}}
For a random sample $X_1,\ldots,X_n$ generated under model (\ref{univ_TW}) we have that
$$FD_{X_1,\ldots,X_n}(X_i) = FD_{\lambda\circ h_1,\ldots,\lambda\circ h_n}(\lambda\circ h_i) = FD_{h_1,\ldots,h_n}(h_i),\quad i=1,\ldots,n$$
for a strictly increasing function $\lambda$, because of 1). Then $i^\star := \arg\max_{i} FD_{X_1,\ldots,X_n}(X_i) = \arg\max_{i} FD_{h_1,\ldots,h_n}(h_i)$ and
$$\hat{m}^{FD}_{X_{1:n}} = X_{i^\star} =  \lambda \circ h_{i^\star}= \lambda \circ \hat{m}^{FD}_{h_{1:n}}.$$
Now, because of 2), the continuous mapping theorem, and the identifiability assumption $m_H^{FD}=id_{\I}$ on model (\ref{univ_TW}), we have that under the required conditions on $P_X$
$$\hat{m}^{MBD}_{X_{1:n}} \stackrel{a.s.}{\longrightarrow} \lambda \circ m^{FD}_H = \lambda \qquad \mbox{ as } \quad n\longrightarrow \infty.\vspace{-1.2cm}$$
\hfill$\blacksquare$ 

\noindent\emph{Proof of Proposition \ref{prop_univ_gen}.}
For a random sample $X_1,\ldots,X_n$ generated under model (\ref{univ_TW}) we have that
$$FD_{T(X_1),\ldots,T(X_n)}(T(X_i)) = FD_{T(\lambda)\circ h_1,\ldots,T(\lambda)\circ h_n}(T(\lambda)\circ h_i) = FD_{h_1,\ldots,h_n}(h_i),\quad i=1,\ldots,n$$
where the first equality holds since $T$ preserves warping functions and the second one because $T(\lambda)$ is a strictly monotone function and $FD$ satisfies 1) in Prop. \ref{prop_univ}. Then $i^\star := \arg\max_{i} FD_{T(X_1),\ldots,T(X_n)}(T(X_i)) = \arg\max_{i} FD_{h_1,\ldots,h_n}(h_i)$ and
$$\hat{m}^{FD}_{T(X_{1:n})} = T(X_{i^\star}) = T(\lambda) \circ h_{i^\star}= T(\lambda) \circ \hat{m}^{FD}_{h_{1:n}}$$
so
$$\arg\max_{X_i} FD_{T(X_1),\ldots,T(X_n)}(T(X_i)) = X_{i^\star} = \lambda \circ h_{i^\star} = \lambda \circ \hat{m}^{FD}_{h_{1:n}}.$$
Now, because of 2) in Prop. \ref{prop_univ}, the continuous mapping theorem, and the identifiability assumption $m_H^{FD}(t) =id_{\I}$ on model (\ref{univ_TW}), we have that under the required conditions on $P_X$\\
$$\arg\max_{X_i} FD_{T(X_1),\ldots,T(X_n)}(T(X_i)) \stackrel{a.s.}{\longrightarrow} \lambda \circ m^{FD}_H = \lambda \qquad \mbox{ as } \quad n\longrightarrow \infty.\vspace{-1.2cm}$$
\hfill$\blacksquare$ 

\noindent\emph{Proof of Proposition \ref{composition}.} For $t\in \I$, $H(t)$ is an absolutely continuous random variable on $\I$ with central symmetry with respect to $t$, i.e., $f_{H(t)}(x) = f_{H(t)}(2t-x)$, $\forall x\in \I$, where $f_{H(t)}$ denotes the density function of the marginal distribution $H(t)$. Let $(t-k_t,t+k_t)$ define its support, for some $k_t>0$.\\
Also, for $t\in \I$, $\Psi^{-1}(t)$ is a random variable on $\I$ with central symmetry around $t$, i.e., $P(\Psi^{-1}(t) \leq x) = P(\Psi^{-1}(t) \geq 2t-x)$, $\forall x\in \I$.\\
Then, the marginal distribution function of $G(t)= \Psi(H(t))$ at $t\in\I$ is given by
$$ P(G(t) \leq t) = \int_{t-k_t}^{t+k_t} P(\Psi(u)\leq t) f_{H_t}(u) du = \int_{t-k_t}^{t+k_t} P(\Psi(2t-u)\geq t) f_{H_t}(2t-u) du$$
$$= \int_{t-k_t}^{t+k_t} P(\Psi(v)\geq t) f_{H_t}(v) dv = P(G(t) \geq t)$$
where the second equality holds because of the central simetry of $\Psi^{-1}(t)$ and $H(t)$ around $t$, and the third equality follows a change of variable.\\
Therefore, for all $t\in \I$, $P(G(t) \leq t) =P(G(t) \geq t) \geq 1/2$, that is, $t$ is the median of the marginal distribution of $G(t)$. Now, following Theorem 3 in \cite{QID}, any quantile integrated depth on $G$ will be uniquely maximized at $id_{\I}(t) = t$, $t\in\I$. That is, $m_G^{FD} = id_\I$, for such a functional depth measure $FD$.\hfill
$\blacksquare $

\subsection*{A2. Results for simulations under outlier contamination}
Tables \ref{res_outliers} and  \ref{res_outliers2} present the results of the simulation study under shape outlier contamination, with contamination rates $c=0$, $0.05$ and $0.1$. Figure \ref{figs_simdata_outs} illustrates four particular simulation runs under different settings.
\begin{figure}[h!]
\begin{center}
\hspace*{-1cm}\includegraphics[width=15.5cm]{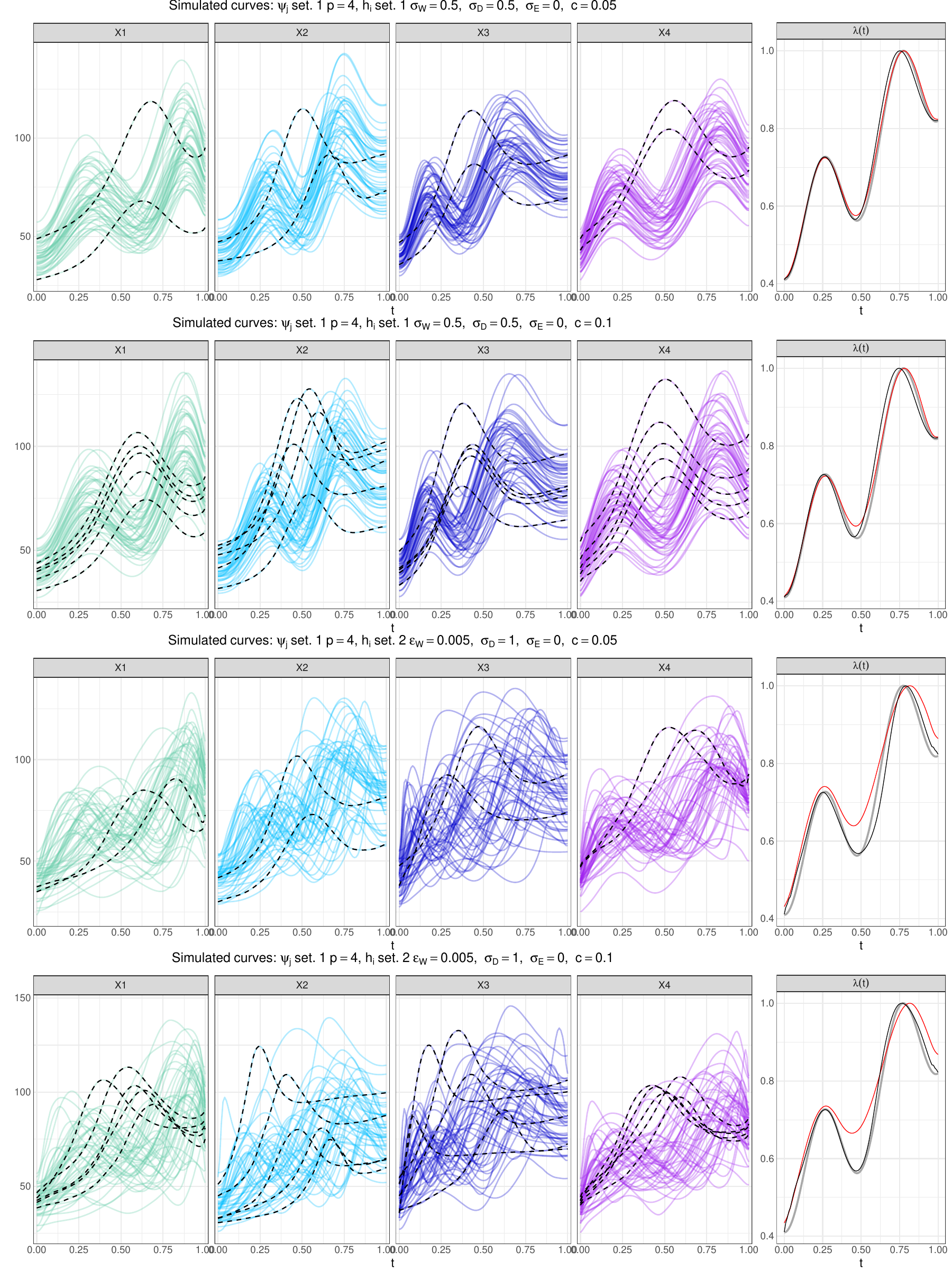}\end{center}
\caption{Four simulated data sets under model (\ref{sim_model}) with contamination of $100c\%$ of the curves ($c=0.05$ or $0.1$), according to (\ref{lamb_cont}). Outlying observations are presented with black dashed lines. The fifth plot on each row represents the target function, $\lambda(t)$, as a thick gray solid line, together with the depth-based and the CM estimates (black and red lines, respectively).}\label{figs_simdata_outs}
\end{figure}

\begin{table}[h!]
\hspace*{-1cm}{\bf Simulations with outliers.} $\mathbf{\psi}$ {\bf setting 1,} $\mathbf{h_i}$ {\bf setting 1.} $\quad\mathbf{n=50}, \quad\mathbf{p=4}$ \vspace{0.3cm}\\
\tiny{
\begin{tabular}{cc|cc|cc|cc}
    {\bf $\mathbf{h_i}$ setting 1}&&\multicolumn{6}{l}{{\bf Additive noise: $\sigma_E=0$. Contamination rate: $c=0$}}\\
    \hline
  Warping & Nuisance  &\multicolumn{2}{l|}{LISE $\times 10^3$}&\multicolumn{2}{l|}{HMISE$\times 10^3$}&\multicolumn{2}{l}{XMISE}\\
   distortion & distortion & $\hat{\lambda}$& $\hat{\lambda}_{CM}$& $\hat{h}_i$&$\hat{h}_{i,CM}$& $\hat{X}_{ij}$&$\hat{X}_{ij,CM}$\\
    \hline
    &$\sigma_D\!=\!0$&2.04  ( 0.63 )&\bf{ 0.06  ( 0.04 )}&\bf{ 0.02  ( 0.02 )}&0.05  ( 0.04 )&\bf{ 0.11  ( 0.03 )}&0.28  ( 0.07 )\\
    \cline{2-8} 
    $\sigma_W\!=\!0.5$&$\sigma_D\!=\!0.5$&0.91  ( 0.8 )&\bf{ 0.08  ( 0.05 )}&\bf{ 0.13  ( 0.03 )}&0.21  ( 0.06 )&\bf{ 6.7  ( 0.79 )}&6.97  ( 0.78 )\\
    \cline{2-8} 
    &$\sigma_D\!=\!1$&\bf{ 0.32  ( 0.4 )}&0.99  ( 0.47 )&1.73  ( 0.31 )&\bf{ 1.54  ( 0.26 )}&\bf{ 74.97  ( 7.25 )}&78.77  ( 8.48 )\\
     \hline
    &$\sigma_D\!=\!0$&\bf{ 0.5  ( 0.88 )}&1.15  ( 0.94 )&\bf{ 0.39  ( 0.47 )}&1.22  ( 0.71 )&\bf{ 1.65  ( 1.03 )}&16.59  ( 8.54 )\\
     \cline{2-8} 
   $\sigma_W\!=\!1$&$\sigma_D\!=\!0.5$&\bf{ 0.38  ( 0.42 )}&1.16  ( 0.8 )&\bf{ 0.4  ( 0.32 )}&1.28  ( 0.61 )&\bf{ 8.02  ( 1.19 )}&23.41  ( 8.63 )\\
    \cline{2-8} 
   &$\sigma_D\!=\!1$&\bf{ 0.86  ( 0.97 )}&3.17  ( 1.1 )&\bf{ 1.84  ( 0.49 )}&2.84  ( 0.66 )&\bf{ 71.67  ( 6.25 )}&100.16  ( 11.93 )\\
    \hline
    \end{tabular}}\vspace{0.3cm}\\
\tiny{
 \begin{tabular}{cc|cc|cc|cc}
     \bf{ $\mathbf{h_i}$ setting 1}&&\multicolumn{6}{l}{\bf{ Additive noise: $\sigma_E=0$. Contamination rate: $c=0.05$}}\\
    \hline
  Warping & Nuisance  &\multicolumn{2}{l|}{LISE $\times 10^3$}&\multicolumn{2}{l|}{HMISE$\times 10^3$}&\multicolumn{2}{l}{XMISE}\\
   distortion & distortion & $\hat{\lambda}$& $\hat{\lambda}_{CM}$& $\hat{h}_i$&$\hat{h}_{i,CM}$& $\hat{X}_{ij}$&$\hat{X}_{ij,CM}$\\
    \hline
    &$\sigma_D\!=\!0$&2.66  ( 0.95 )&\bf{ 0.11  ( 0.06 )}&\bf{ 0.3  ( 0.02 )}&1.3  ( 0.07 )&8.27  ( 0.89 )&\bf{ 3.36  ( 0.36 )}\\
    \cline{2-8} 
    $\sigma_W\!=\!0.5$&$\sigma_D\!=\!0.5$&1.08  ( 0.94 )&\bf{ 0.12  ( 0.08 )}&\bf{ 0.42  ( 0.06 )}&1.48  ( 0.16 )&14.57  ( 1.37 )&\bf{ 9.57  ( 1.01 )}\\
    \cline{2-8} 
    &$\sigma_D\!=\!1$&\bf{ 0.24  ( 0.26 )}&1.08  ( 0.42 )&\bf{ 1.91  ( 0.33 )}&2.06  ( 0.34 )&80.37  ( 6.99 )&\bf{ 79.97  ( 8.53 )}\\
     \hline
    &$\sigma_D\!=\!0$&\bf{ 0.51  ( 0.71 )}&1.29  ( 0.83 )&\bf{ 0.63  ( 0.3 )}&2.23  ( 0.66 )&\bf{ 9.21  ( 1.3 )}&23.08  ( 8.7 )\\
     \cline{2-8} 
   $\sigma_W\!=\!1$&$\sigma_D\!=\!0.5$&\bf{ 0.39  ( 0.36 )}&1.33  ( 0.9 )&\bf{ 0.69  ( 0.29 )}&2.25  ( 0.67 )&\bf{ 15.61  ( 1.61 )}&30.59  ( 10.57 )\\
    \cline{2-8} 
   &$\sigma_D\!=\!1$&\bf{ 0.62  ( 0.58 )}&3.4  ( 1.29 )&\bf{ 2.04  ( 0.41 )}&3.53  ( 0.74 )&\bf{ 77.69  ( 6.81 )}&104.87  ( 12.53 )\\
    \hline
    \end{tabular}}\vspace{0.3cm}\\
\tiny{
\begin{tabular}{cc|cc|cc|cc}
    \bf{ $\mathbf{h_i}$ setting 1}&&\multicolumn{6}{l}{\bf{ Additive noise: $\sigma_E=0$. Contamination rate: $c=0.1$}}\\    \hline
  Warping & Nuisance  &\multicolumn{2}{l|}{LISE $\times 10^3$}&\multicolumn{2}{l|}{HMISE$\times 10^3$}&\multicolumn{2}{l}{XMISE}\\
   distortion & distortion & $\hat{\lambda}$& $\hat{\lambda}_{CM}$& $\hat{h}_i$&$\hat{h}_{i,CM}$& $\hat{X}_{ij}$&$\hat{X}_{ij,CM}$\\
    \hline
    &$\sigma_D\!=\!0$ &3.19  ( 1.27 )&\bf{ 0.27  ( 0.07 )}&\bf{ 0.74  ( 0.04 )}&3.12  ( 0.11 )&19.82  ( 1.35 )&\bf{ 9.3  ( 0.93 )}\\
     \cline{2-8}
    $\sigma_W\!=\!0.5$&$\sigma_D\!=\!0.5$&1.23  ( 1.03 )&\bf{ 0.27  ( 0.1 )}&\bf{ 0.86  ( 0.09 )}&3.16  ( 0.2 )&26.09  ( 1.82 )&\bf{ 14.38  ( 1.23 )}\\
     \cline{2-8}
    &$\sigma_D\!=\!1$&\bf{ 0.29  ( 0.35 )}&1.42  ( 0.56 )&\bf{ 2.21  ( 0.34 )}&2.69  ( 0.38 )&89.41  ( 5.94 )&\bf{ 83.71  ( 7.4 )}\\
     \hline
    &$\sigma_D\!=\!0$&\bf{ 0.46  ( 0.61 )}&1.59  ( 0.93 )&\bf{ 1.02  ( 0.33 )}&3.42  ( 0.78 )&\bf{ 20.56  ( 1.66 )}&32.72  ( 10.48 )\\
     \cline{2-8}
   $\sigma_W\!=\!1$&$\sigma_D\!=\!0.5$&\bf{ 0.5  ( 0.66 )}&1.8  ( 0.96 )&\bf{ 1.14  ( 0.33 )}&3.31  ( 0.81 )&\bf{ 26.81  ( 2 )}&42.68  ( 12.94 )\\
    \cline{2-8}
   &$\sigma_D\!=\!1$&\bf{ 0.62  ( 0.82 )}&4.11  ( 1.43 )&\bf{ 2.41  ( 0.52 )}&3.99  ( 0.73 )&\bf{ 86.48  ( 6.5 )}&114.29  ( 15.35 )\\
    \hline
     \end{tabular}}
    \caption{LISE, HMISE and XMISE means and standard deviations over $N=100$ simulation runs for $p=4$ deterministic component distortion functions ($\psi$ setting 1), individual warping functions generated under setting 1, and $n=50$, with different level of outlier contamination. For easy reference, we also include the case $c=0$, which replicate the results of Table \ref{res1} (cases for which $\sigma_E=0$.)}\label{res_outliers}
\end{table}    
       
\begin{table}[h!]
\hspace*{-1cm}{\bf Simulations with outliers.} $\mathbf{\psi}$ {\bf setting 1,} $\mathbf{h_i}$ {\bf setting 2.} $\quad\mathbf{n=50}, \quad\mathbf{p=4}$ \vspace{0.3cm}\\       
        \tiny{
\begin{tabular}{cc|cc|cc|cc}
   \bf{ $\mathbf{h_i}$ setting 2}&&\multicolumn{6}{l}{{\bf Additive noise: $\sigma_E=0$. Contamination rate: $c=0$}}\\    \hline
  Warping & Nuisance  &\multicolumn{2}{l|}{LISE $\times 10^3$}&\multicolumn{2}{l|}{HMISE$\times 10^3$}&\multicolumn{2}{l}{XMISE}\\
   distortion & distortion & $\hat{\lambda}$& $\hat{\lambda}_{CM}$& $\hat{h}_i$&$\hat{h}_{i,CM}$& $\hat{X}_{ij}$&$\hat{X}_{ij,CM}$\\
    \hline
    &$\sigma_D\!=\!0$&0.85  ( 0.65 )&\bf{ 0.82  ( 0.47 )}&\bf{ 0.41  ( 0.28 )}&0.81  ( 0.4 )&\bf{ 2.89  ( 2.21 )}&10.93  ( 5.69 )\\
     \cline{2-8}
    $\varepsilon_W\!=\!0.005$&$\sigma_D\!=\!0.5$&1.38  ( 1.04 )&1.01  ( 0.62 )&\bf{ 0.49  ( 0.28 )}&1.02  ( 0.45 )&\bf{ 10.28  ( 1.93 )}&17.59  ( 4.45 )\\
     \cline{2-8}
    &$\sigma_D\!=\!1$&\bf{ 1.38  ( 0.87 )}&2.52  ( 1.03 )&\bf{ 1.96  ( 0.46 )}&2.43  ( 0.55 )&\bf{ 72.13  ( 5.71 )}&88.62  ( 9.25 )\\
     \hline
    &$\sigma_D\!=\!0$&\bf{ 1.93  ( 1.04 )}&3.28  ( 1.41 )&\bf{ 1.14  ( 0.58 )}&4.13  ( 1.5 )&\bf{ 9.87  ( 5.36 )}&53.25  ( 19.87 )\\
     \cline{2-8}
   $\varepsilon_W\!=\!0.0075$&$\sigma_D\!=\!0.5$&\bf{ 1.91  ( 1.01 )}&3.39  ( 1.38 )&\bf{ 1.19  ( 0.57 )}&4.35  ( 1.36 )&\bf{ 16.27  ( 3.68 )}&60.3  ( 18.13 )\\
    \cline{2-8}
   &$\sigma_D\!=\!1$&\bf{ 2.33  ( 1.44 )}&5.87  ( 1.79 )&\bf{ 2.61  ( 0.6 )}&5.83  ( 1.4 )&\bf{ 73.26  ( 7.98 )}&118.46  ( 17.85 )\\
    \hline
    \end{tabular}}\vspace{0.3cm}\\
    \tiny{
\begin{tabular}{cc|cc|cc|cc}
   \bf{ $\mathbf{h_i}$ setting 2}&&\multicolumn{6}{l}{\bf{ Additive noise: $\sigma_E=0$. Contamination rate: $c=0.05$}}\\    \hline
  Warping & Nuisance  &\multicolumn{2}{l|}{LISE $\times 10^3$}&\multicolumn{2}{l|}{HMISE$\times 10^3$}&\multicolumn{2}{l}{XMISE}\\
   distortion & distortion & $\hat{\lambda}$& $\hat{\lambda}_{CM}$& $\hat{h}_i$&$\hat{h}_{i,CM}$& $\hat{X}_{ij}$&$\hat{X}_{ij,CM}$\\
    \hline
    &$\sigma_D\!=\!0$&\bf 0.81  ( 0.69 )&0.91  ( 0.53 )&\bf 0.71  ( 0.29 )&1.73  ( 0.49 )&\bf 10.93  ( 2.21 )&16.35  ( 4.7 )\\
     \cline{2-8}
    $\varepsilon_W\!=\!0.005$&$\sigma_D\!=\!0.5$&\bf 0.79  ( 0.59 )&0.93  ( 0.46 )&\bf 0.75  ( 0.21 )&1.91  ( 0.47 )&\bf 17.88  ( 2.13 )&23.18  ( 6.3 )\\
     \cline{2-8}
    &$\sigma_D\!=\!1$&\bf 1.14  ( 0.99 )&2.92  ( 0.93 )&\bf 2.2  ( 0.45 )&2.99  ( 0.57 )&\bf 78.94  ( 6.8 )&96.69  ( 12.63 )\\
     \hline
    &$\sigma_D\!=\!0$&\bf 1.93  ( 1.04 )&3.28  ( 1.41 )&\bf 1.14  ( 0.58 )&4.13  ( 1.5 )&\bf 9.87  ( 5.36 )&53.25  ( 19.87 )\\
     \cline{2-8}
   $\varepsilon_W\!=\!0.0075$&$\sigma_D\!=\!0.5$&\bf 2  ( 1.22 )&3.92  ( 1.56 )&\bf 1.55  ( 0.6 )&5.14  ( 1.61 )&\bf 23.83  ( 4.17 )&69.57  ( 19.72 )\\
    \cline{2-8}
   &$\sigma_D\!=\!1$&\bf 2.38  ( 1.49 )&6.12  ( 1.87 )&\bf 2.88  ( 0.71 )&6.55  ( 1.47 )&\bf 79.4  ( 7.58 )&127.26  ( 21.32 )\\
    \hline
    \end{tabular}}\vspace{0.3cm}\\
    \tiny{
\begin{tabular}{cc|cc|cc|cc}
   \bf{ $\mathbf{h_i}$ setting 2}&&\multicolumn{6}{l}{\bf{ Additive noise: $\sigma_E=0$. Contamination rate: $c=0.1$}}\\    \hline
  Warping & Nuisance  &\multicolumn{2}{l|}{LISE $\times 10^3$}&\multicolumn{2}{l|}{HMISE$\times 10^3$}&\multicolumn{2}{l}{XMISE}\\
   distortion & distortion & $\hat{\lambda}$& $\hat{\lambda}_{CM}$& $\hat{h}_i$&$\hat{h}_{i,CM}$& $\hat{X}_{ij}$&$\hat{X}_{ij,CM}$\\
    \hline   
    &$\sigma_D\!=\!0$& \bf 0.84  ( 0.76 )&1.25  ( 0.6 )&\bf 1.11  ( 0.32 )&3.03  ( 0.58 )&\bf 22.6  ( 3.2 )&29.09  ( 8.92 )\\
     \cline{2-8}
    $\varepsilon_W\!=\!0.005$&$\sigma_D\!=\!0.5$&\bf 0.92  ( 0.71 )&1.37  ( 0.62 )&\bf 1.25  ( 0.36 )&2.99  ( 0.56 )&\bf 29.03  ( 2.9 )&35.19  ( 9.36 )\\
     \cline{2-8}
    &$\sigma_D\!=\!1$&\bf 1.17  ( 0.86 )&3.59  ( 1.06 )&\bf2.63  ( 0.53 )&3.87  ( 0.77 )&\bf88.32  ( 7 )&107  ( 12.81 )\\
     \hline
   &$\sigma_D\!=\!0$&\bf1.93  ( 1.31 )&4.49  ( 1.89 )&\bf2.05  ( 0.73 )&6.41  ( 1.6 )&\bf29.26  ( 5.75 )&80.92  ( 23.59 )\\
     \cline{2-8}
   $\varepsilon_W\!=\!0.0075$&$\sigma_D\!=\!0.5$&\bf1.92  ( 1.24 )&4.69  ( 1.75 )&\bf1.96  ( 0.6 )&6.13  ( 1.59 )&\bf34.93  ( 4.35 )&86.22  ( 22.61 )\\
    \cline{2-8}
   &$\sigma_D\!=\!1$&\bf 1.91  ( 1.16 )&7.3  ( 2.14 )&\bf 3.26  ( 0.7 )&7.18  ( 1.41 )&\bf88.83  ( 7.49 )&141.03  ( 20.74 )\\
    \hline
    \end{tabular}}
    \caption{LISE, HMISE and XMISE means and standard deviations over $N=100$ simulation runs for $p=4$ deterministic component distortion functions ($\psi$ setting 1), individual warping functions generated under setting 2,  and $n=50$, with different level of outlier contamination. For easy reference, we also include the case $c=0$, which replicate the results of Table \ref{res1} (cases for which $\sigma_E=0$.)}\label{res_outliers2}
\end{table}

\section*{Data availability statement}
Data derived from public domain resources: The data sets used in Section \ref{appli} are available in Eumetsat User Portal at \url{https://doi.org/10.15770/EUM_SAF_OSI_0022} and in Eurostat Data Portal at \url{https://doi.org/10.2908/DEMO_FORDAGEC}.

\bibliographystyle{plainnat}
\bibliography{Biblio}

\end{document}